\documentclass[11pt]{article}
\usepackage{vmargin}
\setmarginsrb{1in}{1in}{1in}{1in}{0mm}{0mm}{0mm}{7mm}

\usepackage{amssymb,amsmath,amsthm}

 \usepackage{graphicx}
\usepackage{boxedminipage}

\usepackage{algorithm}
\usepackage{algorithmic}

\newtheorem{redrule}{Reduction Rule}	
\newtheorem{theorem}{Theorem}
\newtheorem{observation}{Observation}
\newtheorem{proposition}{Proposition}
\newtheorem{corollary}{Corollary}
\newtheorem{lemma}{Lemma}
\theoremstyle{remark}

\theoremstyle{definition}
\newtheorem{define}{Definition}

\theoremstyle{remark}

\newcommand{\name}[1]{\textsc{#1}}

\newcommand{\jc}{$\theta_c$}

\newcommand{\tw}{{\mathbf{tw}}}

\newcommand{\pmin}{{\sc $p$-min-MSO}}

\newcommand{\h}[1]{\end{document}}

\newcommand{\qel}{\mbox{$q$--Expansion Lemma}}
\newcommand{\fd}{$p$-\name{$\mathcal{F}$-Deletion}}
\newcommand{\tfd}{$p$-\name{\jc{}-Deletion}}

\newcommand{\dhs}{\name{Diamond Hitting Set}}

\title{Hitting forbidden minors: Approximation and Kernelization}
\date{}

 \author{Fedor V. Fomin\thanks{Department of Informatics, University of Bergen, N-5020 Bergen, Norway.  
\texttt{fomin@ii.uib.no}.} 
\and Daniel Lokshtanov\thanks{
University of California, San Diego, 
 La Jolla, CA 92093-0404, 
USA,  \texttt{dlokshtanov@cs.ucsd.edu}}
\and  Neeldhara Misra\thanks{The Institute of Mathematical Sciences, Chennai - 600113, India.  
 \texttt{\{neeldhara|gphilip|saket\}@imsc.res.in}.}
\and \addtocounter{footnote}{-1} Geevarghese Philip\footnotemark 
\and \addtocounter{footnote}{-1} Saket Saurabh\footnotemark}
 
\begin{document}

\pagebreak 

\maketitle

\begin{abstract}

We study a general class of problems called \fd{}
problems. In an \fd{} problem, we 
are asked whether a subset   of at most $k$ vertices can be deleted from a graph 
$G$ such that the resulting graph does not contain as a minor any graph from the 
family ${\cal F}$ of forbidden minors.
   We obtain a number of algorithmic results on the \fd{}~problem when  $\mathcal{F}$ contains a  planar graph. We give
\begin{itemize}
\setlength{\itemsep}{-2pt}
\item a linear vertex kernel on graphs excluding $t$-claw $K_{1,t}$,  the star with  $t$ leves,  as an induced subgraph, where  $t$ is a fixed integer.
 
\item an approximation algorithm achieving  an approximation ratio of 
$O(\log^{3/2} OPT)$, where $OPT$ is the size of an optimal solution on general undirected graphs. 
\end{itemize}
Finally, we obtain polynomial kernels for the case when  $\cal F$ contains  graph  \jc{} 
as a minor for a fixed integer $c$. The graph \jc{} consists of two
vertices connected by $c$ parallel edges. Even though this may appear to be a very restricted class
of problems it already encompasses well-studied problems such as {\sc Vertex Cover}, {\sc Feedback Vertex Set} and \dhs{}. 
The generic kernelization algorithm is based on a non-trivial application of protrusion techniques, previously  used only  for problems on
topological graph classes.
 
\end{abstract}
 
\section{Introduction}
 Let  $\cal F$ be a finite set of graphs. In an \fd{} problem\footnote{We use prefix $p$  to distinguish the  parameterized version of the problem.}, we are given an $n$-vertex graph
$G$ and an integer $k$ as input, and asked whether at most $k$ vertices
can be deleted from $G$ such that the resulting graph does not contain a graph from   ${\cal
F}$ as a minor.  More precisely the problem is defined as follows.

\begin{center}

\begin{boxedminipage}{.96\textwidth}

 \fd{}

\begin{tabular}{ r l }
\textit{~~~~Instance:} & A graph $G$  and a non-negative integer $k$. \\
\textit{Parameter:} & $k$\\
\textit{Question:} & Does there exist $S \subseteq V(G)$, $|S| \leq k$, \\
~~~~~~~~~~~~~~~~~~ & such that $G \setminus S$ contains no graph from ${\cal F}$ as a minor?
 \\
\end{tabular}

\end{boxedminipage}

\end{center}
We refer to such subset $S$ as $\cal F$-hitting set.   The \fd{}~problem is a generalization of several fundamental problems. For example, when  ${\cal
F}=\{K_2\}$, a complete graph on two vertices, this is the 
{\sc Vertex Cover} problem. When ${\cal
F}=\{C_3\}$,   a cycle on three vertices, this is the 
   {\sc Feedback Vertex Set} problem.  Another famous cases are  ${\cal
F} =\{K_{2,3}, K_4\}$, ${\cal
F} =\{K_{3,3}, K_5\}$ and ${\cal
F} =\{K_{3}, T_2\}$, which correspond to removing vertices to obtain outerplanar graphs, planar 
graphs and graphs of pathwidth one respectively.  Here $K_{i,j}$ is a complete bipartite graph with bipartitions of sizes $i$ and $j$, $K_i$ is a complete graph on $i$ vertices, and $T_2$ is the graph  in the left of Figure~\ref{fig:theta_c}.
In literature these problems are known as   {\sc $p$-Outerplanar Deletion Set}, {\sc $p$-Planar Deletion Set} and 
{\sc $p$-Pathwidth One Deletion Set} respectively. 
 
  \begin{figure}[t]
\begin{center}
\includegraphics[width=0.9\textwidth]{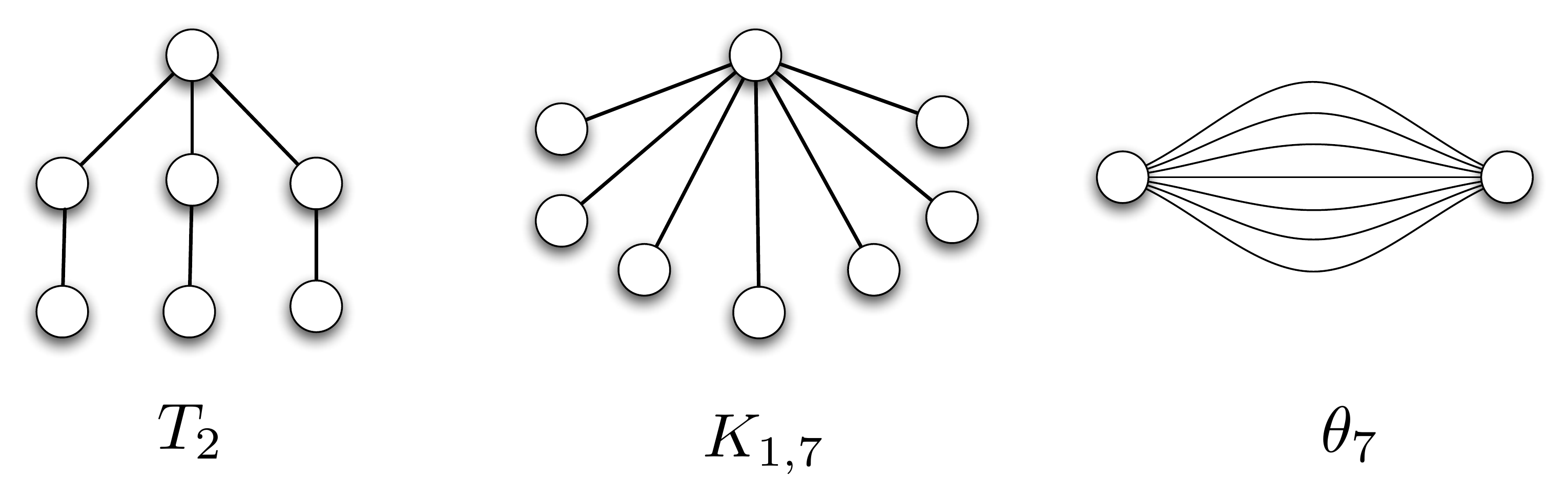}
\caption{\label{fig:theta_c} Graphs $T_2$,  $t$-claw $K_{1,t}$ with $t=7$, and  \jc{} with $c = 7$}
 
\end{center}
 
\end{figure}

{Our interest in the \fd{} problem is motivated  by   its generality and the recent development in  kernelization or polynomial time preprocessing.  The parameterized complexity of this general problem is well understood. 
By a celebrated result of Robertson and Seymour, every \fd{} problem is  fixed-parameter tractable (FPT). That is, 
there is an algorithm solving the problem  in time  $O(f(k) \cdot n^3)$ \cite{RobertsonS13}. In this paper we study this 
problem from the view point of polynomial time preprocessing and approximation, when the obstruction set $\cal F$ satisfies certain 
properties. }

Preprocessing as a strategy for coping with hard problems is universally applied 
in practice and the notion of {\em kernelization} provides a mathematical framework for analyzing the quality of preprocessing strategies. We consider parameterized problems, where every instance $I$ comes with a {\em parameter} $k$. Such a problem is said to admit a {\em polynomial kernel} if every instance $(I,k)$ can be reduced in polynomial time to an equivalent instance with both size and parameter value bounded by a polynomial in $k$.
The study of kernelization is a major research frontier of 
Parameterized Complexity and many important recent advances in the area
are on kernelization. These include 
general results
showing that  certain classes of parameterized problems have polynomial kernels~\cite{Alon:2010vp,H.Bodlaender:2009ng,F.V.Fomin:2010oq,kratsch-snp}. 
The recent development of a framework for ruling out polynomial kernels under
certain complexity-theoretic
assumptions~\cite{BDFH08,Dell:2010sh,FS08}  
 has added a new dimension to
the field and strengthened its connections to classical complexity.  For overviews of the kernelization we  refer to surveys~\cite{Bodlaender09,GN07SIGACT}  and to the corresponding chapters in books on Parameterized Complexity 
\cite{FlumGroheBook,Niedermeierbook06}.

While the initial interest in kernelization was driven mainly by practical applications, the notion of kernelization   appeared to be very important in theory as well. It is well known, see e.g. \cite{DowneyF98}, that a parameterized problem is fixed parameter tractable, or belongs to the class FPT,  if and only if it has (perhaps exponential) kernel. Kernelization is a way to  classify the problems belonging to FPT, the most important class in  Parameterized Complexity, according to the sizes of their kernels.   So far, most of the work done in the field of kernelization  is still specific to particular problems and  powerful unified techniques to identify classes of problems with polynomial kernels are still in nascent stage. 
One of the fundamental challenges in the area is the possibility to characterise     general classes of parameterized problems possessing kernels of polynomial sizes.  From this perspective,  the class of  the \fd{} problems is  very interesting
because it contains as special cases  
{\sc $p$-Vertex Cover}  and   {\sc $p$-Feedback Vertex Set} problems which  are the most  intensively studied 
problems  from the kernelization perspective.

\paragraph{Our contribution and key ideas.} One of the main conceptual contributions of this work is the extension of protrusion techniques, initially  developed in  \cite{H.Bodlaender:2009ng,F.V.Fomin:2010oq} for obtaining meta-kernelization theorems for problems on sparse  graphs like planar and $H$-minor-free graphs, to   general graphs.
We demonstrate   this  by obtaining  a number of   kernelization results  on the \fd{}~problem, when   $\mathcal{F}$ contains a  planar graph. 
Our first result is the following theorem for graphs containing no  star with $t$ leaves $K_{1,t}$, see  Figure~\ref{fig:theta_c}, as an induced subgraph. 

\begin{theorem}
\label{thm:k_1_t}
Let $\cal F$ be  an  obstruction set containing  a planar graph.
   Then  \fd{}   admits  a linear vertex kernel on  graphs excluding $K_{1,t}$ as an induced subgraph, where  $t$ is a fixed integer.
 \end{theorem}  
 Several well studied graph classes do not contain graphs with induced $K_{1,t}$. Of course,  every graph with maximum vertex degree at most $t-1$ is  $K_{1,t}$-free.
  The class of $K_{1,3}$-free graphs, also known as claw-free graphs, contains   line graphs and de Bruijn graphs.  Unit disc graphs are known to be $K_{1,7}$-free \cite{Clark:1990eu}.  We remark that  the number of vertices $O(k)$ in kernels of  
  Theorem~\ref{thm:k_1_t} is (up to a multiplicative constant) optimal, unless P$=$NP. 
  
   Our kernelization is   a divide and conquer algorithm which finds and replaces large protrusions, that is, subgraphs of constant treewidth   separated from the remaining part of the graph by a constant number of vertices,  by smaller, ``equivalent'' protrusions. Here we use the results
from the   work by Bodlaender et al.~\cite{H.Bodlaender:2009ng} that enable this step whenever the parameterized problem in question ``behaves like a regular language''. To prove that \fd{} has the desired properties for this step, we formulate the problem in monadic second order logic and show that it exhibits
certain monotonicity properties.  As a corollary we obtain that {\sc $p$-Feedback Vertex Set}, 
{\sc $p$-Diamond Hitting Set}, {\sc $p$-Pathwidth One Deletion Set}, {\sc $p$-Outerplanar Deletion Set} admit linear 
vertex kernel on graphs excluding $K_{1,t}$ as an induced subgraph. With the same methodology we also obtain 
$O(k \log k)$ vertex kernel for {\sc $p$-Disjoint Cycle Packing} on graphs excluding $K_{1,t}$ as an induced subgraph. 
It is worthwhile to mention that {\sc $p$-Disjoint Cycle Packing} does not admit polynomial kernel on general 
graphs~\cite{BodlaenderThomasseYeo2009}.

 \medskip 
  Let  \jc{} be  
 a graph with two vertices and $c\geq 1$ parallel edges,  see Figure~\ref{fig:theta_c}.
Our second   result is the following theorem on general graphs. 
\begin{theorem}
\label{thm:thetackernel}
Let $\cal F$ be  an  obstruction set containing  \jc{}. Then \fd{} admits a kernel of size $O(k^2 \log^{3/2}k)$. 
 \end{theorem}  
A number of well-studied NP-hard combinatorial problems are special cases of  \tfd{}. 
When $c=1$, this   is the classical 
{\sc Vertex Cover} problem~\cite{NemT74}. For $c=2$, this is another well studied problem, the  {\sc Feedback Vertex Set} problem~\cite{BafnaBF99,BarYGJ98,ChudakGW98,Karp72}. 
When $c=3$, this is the \dhs{} problem~\cite{Fiorini:2009ipco}. Let us note that the size of the best known kernel for $c=2$ is $O(k^2)$, which is very close to the size of the kernel in Theorem~\ref{thm:thetackernel}. Also 
Dell and van Melkebeek  proved  that  no NP-hard vertex deletion problem based on a graph property that is inherited by subgraphs can have kernels of size $O(k^{2-\varepsilon})$   unless $coNP \subseteq NP/poly$ \cite{Dell:2010sh} and thus the sizes of the kernels in Theorem~\ref{thm:thetackernel} are tight up to polylogarithmic factor.

 The proof of Theorem~\ref{thm:thetackernel}  is obtained in a series of non-trivial steps. 
 The very high level idea is to reduce the general case to problem on graphs of bounded degrees, which allows us to   use the protrusion techniques as in the proof of Theorem~\ref{thm:k_1_t}. However, vertex degree reduction is not straightforward and requires several new ideas.   
 One of the new tools is  a generic
    $O(\log^{3/2} OPT)$-approximation algorithm  for the \fd{}~problem when the class of excluded minors for $\mathcal{F}$ contains at least 
one planar graph. More precisely, we  obtain the following result,  which is  interesting in its own.   
\begin{theorem}
\label{thm:fapprox}
Let $\cal F$ be  an  obstruction set containing  a planar graph. Given a graph $G$, in polynomial time 
we can find a subset $S\subseteq V(G)$ such that $G[V\setminus S]$   contains  no   element of $\cal F$ as a minor and $|S|= O(OPT\cdot \log^{3/2} OPT)$. Here $OPT$ is the minimum size of such a set $S$. 
\end{theorem}
While several generic approximation  algorithms were known for  problems of 
minimum vertex deletion to obtain subgraph with property $P$, like when $P$ is a hereditary property   with a finite number of minimal forbidden subgraphs \cite{Lund:1993cl}, or 
 can be expressed as a universal first order sentence over subsets of edges of the graph 
\cite{Kolaitis:1995vf}, we are not aware of any generic approximation algorithm for the case when a property $P$ is characterized by a finite set of forbidden minors.

We then use the approximation algorithm as a subroutine  in a polynomial time algorithm that transforms the input instance $(G,k)$ 
into an equivalent instance $(G',k')$ such that $k'\leq k$ and the maximum degree of $G'$ is bounded by $O(k \log^{3/2} k)$. 
An important combinatorial tool used in designing this algorithm is the~\qel{}. For $q=1$ this lemma is Hall's theorem and its usage can be seen as applying Crown Decomposition technique~\cite{AFLS07,CFJ04}. 
After we manage to reduce the maximum degree of a graph, we  apply protrusion techniques and 
prove  Theorem~\ref{thm:thetackernel}.

\paragraph{Related work.}  
All non-trivial \fd{} problems are
NP-hard~\cite{LewisYannakakis1980}. 
By one of the most well-known
consequences of the celebrated Graph Minor theory of Robertson and
Seymour,  the \fd{}~problem is fixed parameter tractable for every finite set of forbidden minors. A special case of that problem, when the set  $\cal F$ contains  \jc{} 
  was studied from approximation and parameterized perspectives. 
In particular, the case of $p$-$\theta_1$-{\sc  Deletion}
or,  equivalently,   $p$-{\sc Vertex Cover},
 is   the most well-studied problem in  Parameterized Complexity.  
 Different kernelization techniques were tried for it, resulting in a
 $2k$-sized vertex kernel~\cite{AFLS07,ChenKJ01,DFRS04,Hochbaum:1994kl}.
 For the kernelization of {\sc $p$-Feedback Vertex Set}, or
 $p$-$\theta_2$-{\sc   Deletion}, there has been a sequence of
 dramatic improvements starting from an $O(k^{11})$ vertex kernel by
 Buragge et al.~\cite{BEFLMR2006}, improved to  $O(k^3)$ by
 Bodlaender~\cite{Bod07}, and then finally to $O(k^2)$ by
 Thomass\'e~\cite{T09}. Recently Philip et al.~\cite{PhilipRS09} and Cygan et al.~\cite{CyganPPW10} obtained polynomial kernels for 
{\sc $p$-Pathwidth One Deletion Set}.  Constant factor approximation algorithm are known for  {\sc Vertex Cover} and  {\sc  Feedback Vertex Set}~\cite{BafnaBF99,Bar-YehudaE81}. 
Very recently, a constant factor approximation algorithm for the \dhs{} problem, or  $p$-$\theta_3$-{\sc  Deletion}, was obtained in \cite{Fiorini:2009ipco}. 
 Prior to our work, no polynomial kernels were  known for  $p$-\dhs{} or more general families of  \fd{} problems.

\medskip

The remaining part of the  paper is organised as follows. In Section~\ref{sec:preliminaries} we provide preliminaries on  basic notions from Graph Theory and  Logic used in the paper. Section~\ref{sec:kernelizationk1t} is devoted to the proof of  Theorem~\ref{thm:k_1_t}. In Section~\ref{sec:approx} we give an approximation algorithms proving Theorem~\ref{thm:fapprox}. The proof of Theorem~\ref{thm:thetackernel} is given in Section~\ref{sec:kernelization}.  We conclude with open questions in Section~\ref{section:conclision}.

\section{Preliminaries}\label{sec:preliminaries}
In this section we give various definitions which we use in the paper.
For $n\in\mathbb{N}$, we use $[n]$ to denote the set $\{1,\ldots,n\}$. 
We use $V(G)$ to denote the vertex set of a graph $G$, and $E(G)$ to denote 
the edge set. The degree of a vertex $v$ in $G$ is the number of edges
incident on $v$, and is denoted by $d(v)$. We use $\Delta(G)$ to denote the maximum degree of
  $G$.
A graph~$G'$ is a \emph{subgraph} of~$G$ if~$V(G') \subseteq V(G)$
and~$E(G') \subseteq E(G)$.  The subgraph~$G'$ is called an
\emph{induced subgraph} of $G$ if $E(G') = \{\{u,v\} \in E(G) \mid
u,v \in V(G')\}$. Given a subset $S\subseteq V(G)$ the subgraph
induced by $S$ is denoted by~$G[S]$.  The subgraph induced by $
V(G)\setminus S$ is denoted by $G\setminus S$.  We denote by
$N(S)$ the open neighborhood of $S$, i.e. the set of vertices in
$V(G)\setminus S$ adjacent to $S$. Let  $\cal F$ be a finite set of graphs. 
A vertex subset $S\subseteq V(G)$ of a graph $G$ is said to be a 
$\cal{F}$-{\em hitting set} if $G\setminus S$ does not contain any graphs in the family 
$\cal F$ as a minor. 
  

By {\em contracting} an edge $(u,v)$ of a graph $G$, we mean 
identifying the vertices $u$ and $v$, keeping all the parallel edges 
and removing all the loops. A {\em minor} of a graph $G$ is a graph $H$ that can be obtained from 
a subgraph of $G$ by contracting edges. We keep parallel edges after 
contraction since the graph \jc{} which we want to exclude as a minor itself 
contains parallel edges. 

  Let $G,H$ be two graphs. A subgraph $G'$ of $G$ 
is said to be a \emph{minor-model} of $H$ in $G$ if $G'$ contains $H$
as a minor. The subgraph $G'$ is a \emph{minimal minor-model} of $H$ in 
$G$ if no proper subgraph of $G'$ is a minor-model of $H$ in $G$. 

A graph class $\mathcal C$ is {\em  minor closed} 
if any minor of any graph in $\mathcal C$ is also an element of $\mathcal C$. 
A minor closed graph class $\mathcal C$ is $H${\em -minor-free}  or simply $H${\em -free} if $H \notin \mathcal C$.

\subsection{Monadic Second Order Logic (MSO)}
\label{countmsop}
The syntax of MSO on graphs includes the logical connectives $\vee$, $\land$, $\neg$, 
$\Leftrightarrow $,  $\Rightarrow$, variables for 
vertices, edges, sets of vertices and sets of edges, the quantifiers $\forall$, $\exists$ that can be applied 
to these variables, and the following five binary relations: 
\begin{enumerate}

\item $u\in U$ where $u$ is a vertex variable 
and $U$ is a vertex set variable;\item  $d \in D$ where $d$ is an edge variable and $D$ is an edge 
set variable; \item  $\mathbf{inc}(d,u)$, where $d$ is an edge variable,  $u$ is a vertex variable, and the interpretation 
is that the edge $d$ is incident on the vertex $u$; \item $\mathbf{adj}(u,v)$, where  $u$ and $v$ are 
vertex variables $u$, and the interpretation is that $u$ and $v$ are adjacent; \item  equality of variables representing vertices, edges, set of vertices and set of edges. 
\end{enumerate}

Many common graph-theoretic notions such as vertex degree, connectivity,
planarity, being acyclic, and so on, can be expressed in MSO, as can be
seen from introductory
expositions~\cite{BorieParkerTovey1992,HandbookGraphGrammars1997Ch5}. Of
particular interest to us are \pmin{} problems. In a \pmin{} graph
problem $\Pi$, we are given a
graph $G$ and an integer $k$ as input. The objective is to decide
whether there is a vertex/edge set $S$ of size at most $k$ such that the
MSO-expressible predicate $P_\Pi(G,S)$ is satisfied.

\subsection{Parameterized algorithms and Kernels}
\label{paraak}
A parameterized problem $\Pi$ is a subset of $\Gamma^{*}\times \mathbb{N}$ for some finite alphabet $\Gamma$. An instance of a parameterized problem consists of $(x,k)$, where $k$ is called the parameter. A central notion in 
parameterized complexity is {\em fixed parameter tractability (FPT)} which means, for a given instance $(x,k)$, 
solvability in time $f(k)\cdot p(|x|)$, where $f$ is an arbitrary function of $k$ and $p$ is a polynomial in the input size. The notion of {\em kernelization} is formally defined as follows. 

\begin{define} {\rm{[}\textbf{Kernelization, Kernel\rm{]}}~\cite{FlumGroheBook}}
A kernelization algorithm for a parameterized problem $\Pi\subseteq\Sigma^{*}\times\mathbb{N}$
is an algorithm that, given $(x,k)\in\Sigma^{*}\times\mathbb{N}$,
outputs, in time polynomial in $|x|+k$, a pair $(x',k')\in\Sigma^{*}\times\mathbb{N}$
such that (a) $(x,k)\in\Pi$ if and only if $(x',k')\in\Pi$ and (b)
$|x'|,k'\leq g(k)$, where $g$ is some computable function. The output
instance $x'$ is called the kernel, and the function $g$ is referred
to as the size of the kernel. If $g(k)=k^{O(1)}$ then we say that
$\Pi$ admits a polynomial kernel.
\end{define}

\subsection{Tree-width and protrusions}
\label{trewap}

Let $G$ be a graph.  A {\em tree decomposition} of a graph $G$ is a pair $(T,\mathcal{ X}=\{X_{t}\}_{t\in V(T)})$
such that
\begin{itemize}
\setlength\itemsep{-1.2mm}
\item $\cup_{t\in V(T)}{X_t}=V(G)$,
\item for every edge $\{x,y\}\in E(G)$ there is a $t\in V(T)$ such that  $\{x,y\}\subseteq X_{t}$, and 
\item for every  vertex $v\in V(G)$ the subgraph of $T$ induced by the set  $\{t\mid v\in X_{t}\}$ is connected.
\end{itemize}

The {\em width} of a tree decomposition is $\max_{t\in V(T)} |X_t|-1$ and the {\em treewidth} of $G$ 
is the  minimum width over all tree decompositions of $G$. A tree
decomposition  $(T,\mathcal{ X})$ is called a {\em nice tree
decomposition} if $T$ is a tree rooted at some node $r$ where
$X_{r}=\emptyset$, each node of $T$ has at most two children, and each
node is of one of the following kinds:
\begin{enumerate}
\setlength\itemsep{-1.2mm}
\item {\em Introduce node}: a node $t$ that has only one child $t'$ where $X_{t}\supset X_{t'}$ and  $|X_{t}|=|X_{t'}|+1$.
\item {\em Forget node}: a node $t$ that has only one child $t'$  where $X_{t}\subset X_{t'}$ and  $|X_{t}|=|X_{t'}|-1$.
\item {\em Join node}:  a node  $t$ with two children $t_{1}$ and $t_{2}$ such that $X_{t}=X_{t_{1}}=X_{t_{2}}$.
\item {\em Base node}: a node $t$ that is a leaf of $t$, is different than the root, and $X_{t}=\emptyset$. 
\end{enumerate}
Notice that, according to the above definitions, the root $r$ of $T$ is
either a forget node or a join node. It is well known that any tree
decomposition of $G$ can be transformed into a nice tree decomposition
in time $O(|V(G)|+|E(G)|)$ maintaining the same
width~\cite{Kloks1994}. We use $G_t$ to denote the graph induced on the
vertices $\cup_{t'}X_t'$, where $t'$ ranges over all descendants of $t$,
including $t$. We use $H_t$ to denote $G_t[V(G_t)\setminus X_t]$.
%

Given a graph $G$ and $S\subseteq V(G)$, we define $\partial_G(S)$ as
the set of vertices in $S$ that have a neighbor in $V(G)\setminus S$. For a
set $S \subseteq V(G)$ the neighborhood of $S$ is $N_G(S) =
\partial_G(V(G) \setminus S)$. When it is clear from the context, we omit the subscripts. We now define the notion of a \emph{protrusion}.

\begin{define}{\rm [\bf $r$-protrusion\rm{]}} 
Given a graph $G$, we say that a set $X \subseteq V(G)$ is an {\em $r$-protrusion} of $G$ if $\tw(G[X])\leq r$ and $|\partial(X)| \leq r$. 
\end{define}


\subsection{$t$-Boundaried Graphs} 
In this section we define {\em $t$-boundaried graphs} and various
operations on them. Throughout this section, $t$ is an arbitrary
positive integer.
\begin{define}{\rm [\bf $t$-Boundaried Graphs\rm ]}
A $t$-boundaried graph is a graph $G$ with $t$ distinguished vertices, uniquely labeled from $1$ to $t$. The set $\partial(G)$ of labeled vertices is called the boundary of $G$. The vertices in $\partial(G)$ are referred to as boundary vertices or terminals.
\end{define}

For a graph $G$ and a vertex set $S \subseteq V(G)$, we will sometimes consider the graph $G[S]$ as the $|\partial(S)|$-boundaried graph with $\partial(S)$ being the boundary.

\begin{define}{\rm [\bf Gluing by $\oplus$\rm]} Let $G_1$ and $G_2$ be two $t$-boundaried graphs. We denote by $G_1 \oplus G_2$ the $t$-boundaried graph obtained by taking the disjoint union of $G_1$ and $G_2$ and identifying each vertex of $\partial(G_1)$ with the vertex of $\partial(G_2)$ with the same label; that is, we glue them together on the boundaries. In $G_1 \oplus G_2$ there is an edge between two labeled vertices if there is an edge between them in $G_1$ or in $G_2$.
\end{define}

In this paper, $t$-boundaried graphs often come coupled with a vertex
set which represents a partial solution to some optimization problem.
For ease of notation we define ${\cal H}_t$ be to be the set of pairs
$(G,S)$, where $G$ is a $t$-boundaried graph and $S \subseteq V(G)$.


\begin{define}{\rm [\bf Replacement\rm]}
Let $G$ be a graph containing a $r$-protrusion $X$. Let $G_1$ be an $r$-boundaried
graph. The act of replacing $G[X]$ with $G_1$ corresponds to changing $G$
into $G[(V(G) \setminus X) \cup \partial(X)] \oplus G_1$. 
\end{define}

\subsection{Finite Integer Index}
\label{fiind}
\begin{define}{\rm [\bf Canonical Equivalence\rm]}
For a parameterized problem $\Pi$ and two $t$-boundaried graphs $G_1$ and $G_2$, we say that $G_1 \equiv _{\Pi} G_2$ if there exists a constant $c$ such that for all $t$-boundaried graphs $G_3$ and for all $k$,
$$(G_1 \oplus G_3, k) \in \Pi \mbox{ if and only if }  (G_2 \oplus G_3, k+c) \in \Pi.$$
\end{define}
\begin{define}{\rm [\bf Finite Integer Index\rm]}
We say that a parameterized problem $\Pi$ has {\em finite integer index}
if for every $t$ there exists a finite set ${\cal S}$ of $t$-boundaried
graphs such that for any $t$-boundaried graph $G_1$ there exists   $G_2 \in \cal{S}$ such that $G_2 \equiv _{\Pi} G_1$. Such a set ${\cal S}$ is called a set of representatives for $(\Pi, t)$. 
\end{define}
Note that for every $t$, the relation $\equiv_\Pi$ on $t$-boundaried graphs is an equivalence relation. A problem $\Pi$ is finite integer index if and only if for every $t$, $\equiv_\Pi$ is of finite index, that is, has a finite number of equivalence classes. The notion of {\em strong monotonicity} is an easy to check sufficient condition for a \pmin{} problem to have finite integer index.

\begin{define}{\rm [\bf Signatures\rm]}\label{def:signaturemin}
Let $\Pi$ be a \pmin{} problem. For a $t$-boundaried graph $G$ we define
the {\em signature function} $\zeta_G^\Pi : {\cal H}_t \rightarrow
\mathbb{N} \cup \{\infty\}$ as follows. For a pair $(G',S') \in {\cal
H}_t$, if there is no set $S \subseteq V(G)$ ($S\subseteq E(G)$) such
that $P_\Pi(G \oplus G', S \cup S')$ holds, then $\zeta_G^\Pi((G',S')) =
\infty$. Otherwise $\zeta_G^\Pi((G',S'))$ is the size of the smallest $S
\subseteq V(G)$ ($S \subseteq E(G)$) such that $P_\Pi(G \oplus G', S \cup S')$ holds. 
\end{define}

%
\begin{define}{\rm [\bf Strong Monotonicity\rm]}\label{def:minstrongmonmin}
A \pmin{} problem $\Pi$ is said to be \emph{strongly monotone} if there
exists a function $f : \mathbb{N} \rightarrow \mathbb{N}$ such that the
following condition is  satisfied. For every $t$-boundaried graph $G$,
there is a subset $S\subseteq V(G)$ such that for every $(G',S')\in {\cal H}_t$ such that $\zeta_G^\Pi((G',S'))$ is finite, $P_\Pi(G\oplus G',S\cup S')$ holds and $|S|\leq \zeta_G^\Pi((G',S'))+f(t)$. 
\end{define}



\subsection{MSO Formulations}\label{sec:mso-thetac}
We now give
MSO formulations for some properties involving $\cal F$ or \jc{} that we use
in our arguments. For a graph $G$ and a vertex set
$S\subseteq V(G)$,
let $Conn(G,S)$ denote the MSO formula which states that $G[S]$
is connected, and let $MaxConn(G,S)$ denote the MSO formula which
states that $G[S]$ is a maximal connected subgraph of $G$. 

%
%
%
%
%
%

\paragraph{$H$ minor-models. }
Let $\cal F$ be the finite forbidden set. 
For a graph $G$, we use $\phi_{H}(G)$ to denote an MSO formula which
states that $G$ contains $H$ as a minor --- equivalently, that
$G$ contains a minimal $H$ minor model. Let $V(H) = \{h_1, \ldots, h_c\}$. Then, $\phi_{H}(G)$ is given by:

\begin{align}
\label{cmso:jcminor}
\nonumber\phi_{H}(G)\equiv\exists X_1,\ldots,X_c\subseteq V(G)[\\
\nonumber & \bigwedge_{i \neq j} (X_i\cap X_j=\emptyset) \wedge \bigwedge_{1 \leq i \leq c}Conn(G,X_i)\wedge\\
 \nonumber & \bigwedge_{(h_i,h_j) \in E(H)}\exists x \in X_i \wedge y \in X_j [(x,y) \in E(G)]\\
]\end{align}

\paragraph{Minimum-size $\cal F$-hitting set.}

A minimum-size $\cal F$-hitting set of graph $G$ can be expressed as:

\begin{align}
\label{cmso:jchitset}
\mbox{Minimize }S\subseteq V(G)[\bigwedge_{H \in \cal F}\neg\phi_{H}(G \setminus S)]
\end{align}

\paragraph{Largest \jc{} ``flower''.}

Let $v$ be a vertex in a graph $G$. A maximum-size set $M$ of \jc{}
minor-models in $G$, all of which pass through $v$ and no two of
which share any vertex other than $v$, can be represented as:

\begin{align}
\label{cmso:jcflower}
\nonumber \mbox{Maximize }S\subseteq V(G)[\\
\nonumber & \exists F\subseteq E(G)[\forall x\in S[\\
\nonumber \exists X\subseteq V'[MaxConn(G',X) & \wedge x\in X\wedge\forall y\in
S[y\ne x\implies y\notin X]\wedge\phi_{c}(X\cup\{v\})]\\
]]]\end{align}

Here $G'$ is the graph with vertex set $V(G)$ and edge set $F$, and
$V'=V(G)\setminus\{v\}$. $S$ is a system of
distinct representatives for the vertex sets that constitute the elements
of $M$.

\section{Kernelization for  \fd{} on $K_{1,t}$ free graphs}
\label{sec:kernelizationk1t}

In this section we show that  if the obstruction set 
$\cal F$ contains a planar graph then the 
  \fd{} problem has a linear vertex kernel on graphs excluding $K_{1,t}$ as an induced subgraph. We start with the following lemma which is crucial to our kernelization algorithms.  

\begin{lemma}
\label{claim:treewidthbound}%
Let  $\cal F$ be an obstruction set containing a planar graph  of size $h$. If $G$ has a $\cal F$--hitting set of $S$ size at most $k$, then $\tw(G \setminus S)\leq d$  and  $\tw(G)\leq k+d$, where $d=20^{2(14h-24)^5}$.  
\end{lemma}
\begin{proof}
By assumption, $\cal F$ contains at least one 
planar graph.  Let $h$ be the size of the smallest planar graph $H$ contained in ${\cal F}$. By a result of Robertson 
et al.~\cite{RobertsonST94}, $H$ is a minor of the $(\ell \times \ell)$-grid, where $\ell = 14h-24$. 
In the same paper Robertson et al.~\cite{RobertsonST94} have shown that any graph with 
treewidth greater than $20^{2\ell^5}$ contains a $(\ell \times \ell)$-grid as a minor. Let $S$ be a  $\cal F$--hitting set of 
$G$ of size at most $k$. Since the $(\ell \times \ell)$-grid contains $H$ as a minor, we have that $\tw(G \setminus S) \le 20^{2\ell^5}$. 
Therefore, $\tw(G) \le k + d$, where $d = 20^{2\ell^5}$ --- indeed, a tree decomposition of width $(k+d)$ can be obtained by 
adding the vertices of $S$ to every bag in an optimal tree decomposition of $G \setminus S$.  This completes the proof of the lemma. 
\end{proof}

\subsection{The Protrusion Rule --- Reductions Based on Finite Integer Index}
Wo obtain our kernelization algorithm for \fd{} by applying protrusion based reduction rule. That is,  
any large $r$-protrusion for a fixed constant $r$ depending only on $\cal F$ 
(that is, $r$  depends only on the problem) is replaced with a smaller equivalent $r$-protrusion. For this we utilize the following lemma of Bodlaender et al.~\cite{H.Bodlaender:2009ng}.

\begin{lemma}[\cite{H.Bodlaender:2009ng}]\label{lem:red2finiteindex}
Let $\Pi$ be a problem that has finite integer index. Then  there exists a computable function $\gamma : \mathbb{N} \rightarrow \mathbb{N}$ and 
an algorithm that given an instance $(G, k)$ and an $r$-protrusion $X$ of $G$ of size at 
least $\gamma(r)$, runs in $O(|X|)$ time and outputs an instance $(G^*,k^*)$ such that 
$|V(G^*)| < |V(G)|$, $k^* \leq k$, and $(G^*,k^*) \in \Pi$ if and only if $(G,k) \in \Pi$.
\end{lemma}
\noindent
{\bf Remark:} {  Let us remark that if $G$ does not have $K_{1,t}$ as an induced subgraph then the proof of 
Lemma~\ref{lem:red2finiteindex} also ensures that the graph $G'$ does not contain $K_{1,t}$ as an 
induced subgraph. This makes sure that even after replacement we do not leave the graph class we are 
currently working with. The remark is not only true about graphs excluding $K_{1,t}$ as an induced subgraph 
but also for any graph class $\cal G$ that can be characterized by either finite set of forbidden subgraphs or 
induced subgraphs or minors. That is, if   $G$ is in $\cal G$ then so does the  graph $G'$ returned by the 
Lemma~\ref{lem:red2finiteindex}. }

In order to apply Lemma~\ref{lem:red2finiteindex} we need to be
able to efficiently find large $r$-protrusions whenever the
instance considered is large enough. Also, we need to prove that
\fd{} has finite integer index. The next lemma yields a divide
and conquer algorithm for efficiently finding large
$r$-protrusions.

\begin{lemma}
\label{lem:prottfd}
There is a linear time algorithm that given an $n$-vertex graph $G
$ and a set $X \subseteq V(G)$ such that $\tw(G\setminus X) \leq
d$, outputs a $2(d+1)$-protrusion of $G$
of size at least $\frac{n-|X|}{4|N(X)|+1}$. Here $d$ is some constant.
\end{lemma}
\begin{proof}
Let $F = G \setminus X$. The algorithm starts by computing a nice
tree decomposition of $F$ of width at most $d$. Notice that since $d$ is
a constant this can be done in linear time~\cite{Bodlaender96ali}. 
Let $S$ be the vertices in $ V(F)$ that are neighbors of $X$ in $G$, that is, $S=N_G(X)$. 

The nice tree decomposition of $F$ is a pair $(T,\mathcal{B}=\{B_{\ell}\}_{\ell\in V(T)})$, where $T$ is a rooted binary tree. We will now {\em mark} some of the nodes of $T$. For every $v \in S$, we mark the topmost node $\ell$ in $T$ such that $v \in B_\ell$. In this manner, at most $|S|$ nodes are marked. Now we mark more nodes of $T$ by exhaustively applying the following rule: if $u$ and $v$ are marked, mark their least common ancestor in $T$. Let $M$ be the set of all marked nodes of $T$. Standard counting arguments on trees give that $|M| \leq 2|S|$.

Since $T$ is a binary tree, it follows that $T \setminus M$ has at most $2|M|+1$ connected components. Let the vertex sets of these connected components be $C_1, C_2 \ldots C_\eta$, $\eta \leq 2|M|+1$. For every $i \leq \eta$, let $C'_i = N_T(C_i) \cup C_i$ and let $P_i = \bigcup_{u \in C'_i} B_u$. By the construction of $M$, every component of $T \setminus M$ has at most $2$ neighbors in $M$.
Also for every $1\leq i\leq \eta$ and $v\in S$, we have that if $v\in P_i$, then $v$ should be contained in one of the bags of $N_T(C_i)$. In other words,  
 $S \cap P_i \subseteq \bigcup_{u \in C'_i \setminus C_i} B_u$. Thus every $P_i$ is a $2(d+1)$-protrusion of $G$. Since  $\eta \leq 2|M|+1 \leq 4|S|+1$, 
 the pigeon-hole principle yields that there is a protrusion $P_i$ with  at least $\frac{n-|X|}{4|S|+1}$ vertices. The algorithm constructs $M$ and $P_1 \ldots P_\eta$ and outputs the largest protrusion $P_i$. It is easy to implement this procedure to run in linear time. This concludes the proof. 
\end{proof}

No we show that \fd{} has finite integer index. For this we need the following lemma.
\begin{lemma}[\cite{H.Bodlaender:2009ng}]%
\label{lem:stronglymonotone}%
Every strongly monotone \pmin{} problem has finite integer index. 
\end{lemma}

\begin{lemma}
\label{lem:jcisfii} \fd{} has finite integer index.\end{lemma}
\begin{proof}
  One can easily formulate \fd{} in MSO, which shows that it is a  \pmin{} problem, see Section~\ref{sec:mso-thetac}. To complete 
  the proof that \fd{} has finite integer index we show that $\Pi = \fd{}$ is strongly monotone. Given a
  $t$-boundaried graph $G$, with $\partial(G)$ as its boundary,
  let $S''\subseteq V(G)$ be a minimum set of vertices in $G$ such
  that $G\setminus S''$ does not contain any graph in $\cal F$ as
  a minor. Let $S=S''\cup \partial(G)$.

  Now for any $(G',S')\in {\cal H}_t$ such that
  $\zeta_G^\Pi((G',S'))$ is finite,  we have that $G\oplus
  G'[(V(G)\cup V(G'))\setminus (S\cup S')]$ does not contain any
  graph in $\cal F$ as a minor and $|S|\leq
  \zeta_G^\Pi((G',S'))+t$. This proves that \fd{} is strongly
  monotone.  By Lemma~\ref{lem:stronglymonotone}, \fd{} has
  finite integer index.
\end{proof}

\subsection{Analysis and Kernel Size -- Proof of Theorem~\ref{thm:k_1_t}}
Now we give the desired kernel for \fd. We first prove a useful combinatorial lemma.

\begin{lemma}
\label{lem:k1tdegreebound}
Let $G$ be a graph excluding $K_{1,t}$ as an induced subgraph and $S$ be a $\cal F$-hitting set. If $\cal F$ contains 
a planar graph of size $h$, then $|N(S)|\leq g(h,t)\cdot |S|$ for some function $g$ of $h$ and $t$. 
\end{lemma}
\begin{proof}
By Lemma~\ref{claim:treewidthbound}, $\tw(G\setminus S)\leq d$ 
for $d=20^{2(14h-24)^5}$.  It is well known that a graph of 
treewidth $d$ is $d+1$ colorable. Let $v\in S$ and let $S_v$ be its neighbors in $G\setminus S$. 
We first show that $|S_v|\leq (t-1)(d+1)$. Consider the graph 
$G^*=G[S_v]$. Since
$\tw(G\setminus S)\leq d$ we have that $\tw(G^*)\leq d$ and hence $G^*$ is $d+1$ colorable.  Fix a coloring 
$\kappa $ of $G^*$ with $d+1$ colors and let $\eta$ be the size of the  largest color class. Clearly $\eta \geq (|S_v|/d+1)$. 
Since each color class is an independent set,  we have that $\eta \leq (t-1)$, 
else we will get $K_{1,t}$ as an induced subgraph in $G$. This implies that $|S_v|\leq (t-1)(d+1)$. Since $v$ 
was an arbitrary vertex of $S$, we have that $\sum_{v\in S} |S_v|\leq \sum_{v\in S} (t-1)(d+1) \leq |S|\cdot g(h,t)$. Here 
$g(h,t)=(t-1)(20^{2(14h-24)^5}+1)$. Finally the observation that $N(S)=\cup_{v\in S} S_v$, yields the result. 
\end{proof}

Now we are ready to prove Theorem~\ref{thm:k_1_t}.

\begin{proof}[\bf Proof of Theorem~\ref{thm:k_1_t}]
Let  $(G,k)$ be an instance of \fd{} and $h$ be the size of a smallest planar graph in the obstruction set $\cal F$. 
We first apply Theorem~\ref{thm:fapprox} (to be proved in next section), 
an approximation algorithm for \fd{} with factor $O(\log^{3/2} OPT)$, 
and obtain a set $X$ such that $G \setminus X$   contains no  graph in $\cal F$ as a minor. 
If the size of the set $X$ is more than $O(k \log^{3/2} k)$ then we return that $(G,k)$ is a NO-instance to \fd{}. 
This is justified by the approximation guarantee provided by the Theorem~\ref{thm:fapprox}. 

Let $d$ denote the treewidth of the graph after the removal of $X$, that is, $d := \tw(G \setminus S)$. 
Now  we obtain the  kernel in two phases: we first apply the protrusion rule 
selectively (Lemma~\ref{lem:red2finiteindex}) and get a polynomial kernel. Then, we apply the protrusion 
rule exhaustively on the obtained kernel to get a smaller kernel.  This is done in order to reduce the running time complexity of the kernelization algorithm. 
To obtain the kernel we follow the following steps.

\paragraph{\sl Applying the Protrusion Rule.} 
By Lemma~\ref{claim:treewidthbound},   $d \leq 20^{2(14h-24)^5}$. We apply Lemma~\ref{lem:prottfd} and obtain a $2(d+1)$-protrusion $Y$ of $G$ of size at least 
$\frac{|V(G')|-|X|}{4|N(X)|+1}$.  By Lemma~\ref{lem:jcisfii}, \fd{} has finite integer index. Let 
 $\gamma : \mathbb{N} \rightarrow \mathbb{N}$  be the function defined in Lemma~\ref{lem:red2finiteindex}. 
If $\frac{|V(G')|-|X|}{4|N(X)|+1}\geq \gamma(2d+1)$, then using Lemma~\ref{lem:red2finiteindex} we 
replace the  $2(d+1)$-protrusion $Y$  in $G$ and obtain an instance $(G^*,k*)$ such that  $|V(G^*)| < |V(G)|$, $k^* \leq k$, and $(G^*,k^*) $ is a YES-instance of \fd{} if and only if $(G,k) $ is a YES-instance of \fd{} . Recall that $G^*$ also excludes $K_{1,t}$ as an induced subgraph.

Let $(G^*,k^*)$ be a reduced instance with hitting set $X$. In other words, there is no
$(2d+2)$-protrusion of size $\gamma(2d+2)$ in $G^*\setminus X$, and  Protrusion Rule no longer applies. 
We claim  that the 
number of vertices  in this graph is bounded by $O(k \log^{3/2} k)$. Indeed, since 
  we cannot apply  Protrusion Rule, we have that $\frac{|V(G^*)|-|X|}{4|N(X)|+1}\leq \gamma(2d+2)$.  
Because $k^*\leq k$, we have   that 
\begin{eqnarray*}
|V(G^*)|   \leq  \gamma(2d+2)(4|N(X)|+1)+|X|.
\end{eqnarray*}
By Lemma~\ref{lem:k1tdegreebound},   $|N(X)|\leq g(h,d)\cdot |X|$ and thus 
\begin{eqnarray*}
|V(G^*)|   = O(  \gamma(2d+2) \cdot k \log^{3/2} k) 
=  O(k \log^{3/2} k).
\end{eqnarray*}

This gives us a polynomial time algorithm that returns a vertex kernel of size  $O(k \log^{3/2} k)$. 

\medskip 

Now we give a kernel of smaller size.  We would like to replace every  large $(2d+2)$-protrusion in graph 
by a smaller one.  We find a $(2d+2)$-protrusion $Y$  of size at least $\gamma(2d+2)$ by guessing the boundary 
$\partial(Y)$ of size at most $2d+2$. This could be performed in time $k^{O(d)}$.  So let $(G^*,k^*)$ be the reduced instance on which we cannot apply the Protrusion Rule.  
If $G$ is a YES-instance then there is a 
$\cal F$-hitting set $X$ of size at most $k$ such that $\tw(G \setminus X) \leq d$. Now applying the analysis above with this $X$  yields that $|V(G^*)|=O(k)$. 
Hence if the number of vertices in the reduced instance $G^*$, to which we can not apply the Protrusion Rule, 
is more than  $O(k)$  then we return that $G$ is a NO-instance. This concludes the proof of the theorem.
\end{proof}

\begin{corollary}
{\sc $p$-Feedback Vertex Set}, 
{\sc $p$-Diamond Hitting Set}, {\sc $p$-Pathwidth One Deletion Set}, {\sc $p$-Outerplanar Deletion Set} admit linear 
vertex kernel on graphs excluding $K_{1,t}$ as an induced subgraph.
\end{corollary}
The methodology used in proving Theorem~\ref{thm:k_1_t} is not limited to \fd. For example, it is possible  to 
obtain an $O(k \log k)$ vertex kernel on  $K_{1,t}$-free graphs  for {\sc $p$-Disjoint Cycle Packing}, which is for a given  graph $G$ and   positive integer $k$ to determine if there are $k$ vertex disjoint cycles in $G$. 
  It is iteresting to note that  {\sc $p$-Disjoint Cycle Packing}  does not admit a polynomial kernel on general graphs~\cite{BodlaenderThomasseYeo2009}. 
For our kernelization algorithm, we   use the following Erd\H{o}s-P\'osa property~\cite{ErdosPosa1965}: given a positive integer $\ell$ every graph $G$ either has $\ell$ vertex disjoint cycles or there 
exists a set $S\subseteq V(G)$ of size at most $O(\ell \log \ell)$ such that $G\setminus S$ is a forest. So given a graph $G$  and positive integer $k$ we first apply factor $2$ approximation algorithm 
given in~\cite{BafnaBF99} and obtain a set $S$ such that $G\setminus S$ is a forest. If the size of $S$ is more than  
$O(k \log k) $ then we return that $G$  has $k$ vertex disjoint cycles. Else we use the fact that 
{\sc $p$-Disjoint Cycle Packing}~\cite{H.Bodlaender:2009ng} has finite integer index  and apply protrusion reduction rule in $G\setminus S$ to obtain an equivalent instance $(G^*,k^*)$, 
as in Theorem~\ref{thm:k_1_t}.  
The analysis for kernel size used in the proof of Theorem~\ref{thm:k_1_t} 
together with the observation that $\tw(G\setminus S)\leq 1$ shows that if $(G,k)$ is an yes instance then the size of $V(G^*)$ is at most $O(k \log k)$. 

\begin{corollary}
{\sc $p$-Disjoint Cycle Packing} has  $O(k \log k)$ vertex kernel on graphs excluding $K_{1,t}$ as an induced graph. 
\end{corollary}

Next we extend the methods used in this section for obtaining kernels for \fd{}  on  graphs excluding $K_{1,t}$ as an  induced graph to all graphs, though for restricted $\cal F$, 
that is when $\cal F$ is $\theta_c$. However to achieve this we 
need a polynomial time approximation algorithm with a factor polynomial in optimum size and not depending 
on the input size. For an example for our purpose an approximation algorithm with factor $O(\log n)$ is no good. 
Here we obtain an approximation algorithm for \fd{} with a factor $O(\log^{3/2} OPT)$ whenever the 
finite obstruction set $\cal F$ contains a planar graph. Here $OPT$ is the size of a minimum $\cal F$-hitting set. 
This immediately implies a factor $O(\log^{3/2} n)$ algorithm for all the problems that can categorized by \fd. 
We believe this result has its own significance and is of independent interest.

\section{An approximation algorithm for finding a {$\cal F$}--hitting set}
\label{sec:approx}
In this section, we present an $O(\log^{3/2} OPT)$-approximation algorithm for the \fd{}~problem when the finite obstruction set  $\mathcal{F}$ contains at least one planar graph. 
\begin{lemma}
\label{thm:epconstructive}%
There is a polynomial time algorithm that, given a graph $G$ and a positive integer $k$, either reports that $G$ has no $\cal F$-hitting set of size at most $k$ or finds a $\cal F$-hitting set of size at most $O(k \log^{3/2}k)$.
\end{lemma}
\begin{proof}
We begin by introducing some definitions that will be useful for describing our algorithms. First is the notion of a {\em good labeling function}. Given a nice tree
decomposition $(T,\mathcal{ X}=\{X_{t}\}_{t\in V(T)})$ of a graph $G$, a
function $g:~V(T)\rightarrow \mathbb{N}$ is called a {\em good labeling
function} if it satisfies the following properties: \begin{itemize}
\setlength{\itemsep}{-2pt}
\item if $t$ is a base node then $g(t)=0$;  
 \item if $t$ is an introduce node, then $g(t) = g(s)$, where $s$ is the child of $t$;
  \item if $t$ is a join node, then $g(t) = g(s_1) + g(s_2)$, where $s_1$ and $s_2$ are the children of $t$; and 
   \item if $t$ is a forget node, then $g(t) \in \{ g(s), g(s) + 1\}$, where $s$ is the child of $t$. 
\end{itemize}
A {\em max labeling function} $g$ is defined analogously to a good 
labeling function, the only difference being that for a join node $t$, we have
the condition $g(t) = \max\{g(s_1), g(s_2)\}$. We now turn to the approximation algorithm.

\begin{algorithm}[t]
\caption{{\sc Hit-Set-I-}$(G)$}
\label{fig:jchitsetconst2apx}
\begin{algorithmic}[1]
\IF {$\tw(G)\leq d$}
\STATE Find a minimum $\cal F$-hitting set $Y$ of $G$ and return $Y$.
\ENDIF
\STATE Compute an approximate tree decomposition $(T,\mathcal{ X}=\{X_{t}\}_{t\in V(T)})$ of width $\ell$. 
\IF {$\ell> (k+d) \sqrt{\log (k+d)} $, where $d$ is as in Lemma~\ref{claim:treewidthbound}}
\STATE Return that $G$ does not have $\cal F$--hitting set of size at most $k$.
\ENDIF

\STATE Convert $(T,\mathcal{ X}=\{X_{t}\}_{t\in V(T)})$  to a nice tree decomposition of the same width.
\STATE Find a partitioning of vertex set $V(G)$ into $V_1$, $V_2$ and $X$ (a bag corresponding to a node in $T$) 
such that $\tw(G[V_1])=d$ as described in the proof.\\
\STATE Return $\Big(X\bigcup$ {\sc Hit-Set-I-}$(G[V_1])\bigcup$ {\sc Hit-Set-I-}$(G[V_2])\Big).$ 


\end{algorithmic}
\end{algorithm}

Our algorithm has two phases. In the first phase we obtain a $\cal F$-hitting set of size $O(k^2 \sqrt{\log k})$ and in the second phase we use the hitting set obtained in the first phase to get a $\cal F$-hitting set of size $O(k\log^{3/2}k)$. The second phase could be thought of as ``bootstrapping'' where one uses initial solution to a problem to obtain a better solution. 

By assumption we know that $\cal F$ contains at least one 
planar graph.  Let $h$ be the number of vertices in the smallest planar graph $H$ contained in ${\cal F}$. By a result of Robertson 
et al.~\cite{RobertsonST94}, $H$ is a minor of the $(t \times t)$-grid, where $t = 14h-24$. 
Robertson et al.~\cite{RobertsonST94} have also shown that any graph with treewidth greater than 
$20^{2t^5}$ contains a $t\times t$ grid as a minor. In the algorithm we set ${d=20^{2t^5}}$. 

We start off by describing the first phase of the algorithm, see Algorithm~\ref{fig:jchitsetconst2apx}. We start by checking whether a graph $G$ has treewidth at most $d$ (the first step of the algorithm) using the linear time algorithm of 
Bodlaender~\cite{Bodlaender96ali}. If $\tw(G)\leq d$ then we find an optimum  $\cal F$-hitting set of $G$ in linear time using 
a modification of Lemma~\ref{lem:linALgTW}. If the treewidth of the input graph is more than $d$ then  
we find an approximate tree decomposition of width $\ell$ using an algorithm of Feige et al.~\cite{Feige:2008ge} such that 
$\tw(G)\leq \ell \leq d'\tw(G) \sqrt{\log \tw(G)}$ where $d'$ is a fixed constant.

So if $\ell > (k+d)d'\sqrt{\log (k+d)}$ then by Lemma~\ref{claim:treewidthbound}, we know that the size of 
a minimum $\cal F$-hitting set of $G$ is at least $k+1$. Hence from now onwards we assume that   
$\tw(G)\leq \ell \leq (k+d)d'\sqrt{\log (k+d)}$.  
In the next step we convert the given   tree decomposition to a nice tree decomposition of the same 
width in linear time~\cite{Kloks1994}. Given a nice tree decomposition 
$(T,\mathcal{ X}=\{X_{t}\}_{t\in V(T)})$ of $G$, we compute a {\em partial} 
function $\beta: V(T) \rightarrow \mathbb{N}$, defined as $\beta(t) = \tw(H_t)$. Observe that 
$\beta$ is a max labeling function. We compute $\beta$ in a bottom up fashion starting from base nodes and moving 
towards the root. We stop this computation the first time that   
we find a node $t$ such that $\beta(t)=\tw(H_t)=d$. Let $V_1=V(H_t)$, $V_2=V(G)\setminus V_1 \setminus X_t$ and $X=X_t$. 
After this we recursively solve the problem on the graphs induced on $V_1$ and $V_2$. 

Let us assume that $G$ has a $\cal F$-hitting set of size at most $k$. 
We show that in this case the size of the hitting set returned by the algorithm can be bounded by 
$O(k^2\sqrt{\log k})$.   The above recursive procedure can be thought of as a  
rooted binary tree $\cal T$
where at each non-leaf node of the tree the algorithm makes two recursive calls. We will assume that the 
left child of a node of $\cal T$ corresponds to the graph induced on $V_1$ such that the treewidth of $G[V_1]$ is $d$.    
Assuming that the root is at depth $0$ we show that 
the depth of $\cal T$ is bounded by $k$. Let $P=a_0a_1\cdots a_q$ be a longest path from the root to a leaf and let $G_i$ be the graph 
associated with the node $a_i$. Observe that for every $i\in \{0,\ldots,q-1\}$, $a_i$ has a left child, or else $a_i$ cannot be a non-leaf 
node of $\cal T$. Let the graph 
associated with the left child of $a_i$, $i\in \{0,\ldots,q-1\}$, be denoted by $H_i$. Observe that for every $0\leq i<j \leq q-1$,  
$V(H_i)\cap V(H_j)=\emptyset$ and $\tw(H_i)=d$. This implies that every $H_i$ has at least one $H$ minor model and all of these 
are vertex-disjoint. This implies that $q\leq k$ and hence the depth of $\cal T$ is bounded by $k$.  

Let us look at all the subproblems at depth $i$ in the recursion tree $\cal T$.  
Suppose at depth $i$ the induced subgraphs associated with these subproblems are $G[V_i]$, $i\in [\tau]$, where $\tau$ is some positive integer. 
Then observe that for every $i,j\in[\tau]$ and $i\neq j$, we have that $V_i\cap V_j=\emptyset$, there is no edge $(u,v)$ such that 
$u\in V_i$,  $v\in V_j$, and hence $\sum_{i=1}^\tau k_i\leq k$, where $k_i$ is the size of the 
minimum $\cal F$--hitting set of $G[V_i]$.  Furthermore the number of instances at depth $i$ such that it has at least one 
$H$ minor model and hence contributes to the hitting set is at most $k$. Now Lemma~\ref{claim:treewidthbound} together with the factor 
$d' \sqrt{\log \tw(G)}$ approximation algorithm of Feige et al.~\cite{Feige:2008ge} implies 
that the treewidth of every instance is upper bounded by $(k_i+d)d'\sqrt{\log (k_i+d)}$, where $k_i$ is the size of the minimum $\cal F$--hitting set of $G[V_i]$. Hence the total size of the union of sets added to our hitting set at depth $i$ is at most  
\begin{align*}
& \sum_{i=1}^\tau \chi(i) (k_i+d)d'\sqrt{\log (k_i+d)}  \leq  d'(k+d) \sqrt{\log (k+d)}. 
 \end{align*}
 Here $\chi(i)$ is $1$ if $G[V_i]$ contains at least one $H$ minor model and is $0$ otherwise. We have shown that for each 
$i$ the size of the union of the sets added to the hitting set is at most  $d'(k+d) \sqrt{\log (k+d)}$. 
This together with the fact that the depth is at most $k$  implies 
  that the size of the $\cal F$-hitting set is at most  $O(k^2 \sqrt{\log k})$. Hence if the size of the hitting set returned by 
the algorithm is more than $d'(k+d)k \sqrt{\log (k+d)}$ then we return that $G$ has at no $\cal F$-hitting set of size at most $k$. 
Hence when we move to the second 
phase we assume that we have a hitting set of size $O(k^2 \sqrt{\log k})$. This concludes the description of the first phase of the algorithm. 

\begin{algorithm}[t]
\caption{{\sc Hit-Set-II-}$(G,Z)$}
\label{fig:jchitsetconst3}
\begin{algorithmic}[1]

\IF {$\tw(G)\leq d$}
\STATE Find a minimum $\cal F$-hitting set $Y$ of $G$ and Return $Y$.
\ENDIF
\STATE Compute an approximate tree decomposition $(T,\mathcal{ X}=\{X_{t}\}_{t\in V(T)})$ of width $\ell$. 
\STATE Convert it to a nice tree decomposition of $G$. Now compute the function 
$\mu: V(T) \rightarrow \mathbb{N}$, defined as follows: $\mu(t) =|V(H_t)\cap Z|$.
\IF {$(\mu(r) = 0)$}
\STATE Return $\phi$.
\ELSE
\STATE Find the partitioning of the vertex set $V(G)$ into $V_1$, $V_2$ and $X$ (a bag corresponding to a node in $T$) 
as described in Cases $1$ and $2$ of the proof of Theorem~\ref{thm:epconstructive}. 
\ENDIF
\STATE Return $\Big(X\bigcup$ {\sc Hit-Set-II-}$(G[V_1],Z)\bigcup$ {\sc Hit-Set-II-}$(G[V_2],Z)\Big).$

\end{algorithmic}
\end{algorithm}

Now we describe the second phase of the algorithm. Here we are given the hitting set $Z$ of size $O(k^2 \sqrt{\log k})$ 
obtained from the first phase of the algorithm. The algorithm is given in Algorithm~\ref{fig:jchitsetconst3}. 
The new algorithm essentially uses $Z$ to 
define a good labeling function $\mu$ which enables us to argue that the depth of recursion is upper bounded by $O(\log |Z|)$. 
In particular, consider the function $\mu: V(T) \rightarrow \mathbb{N}$, defined as follows: $\mu(t) =|V(H_t)\cap Z|$. Let $k' := \mu(r)$, where $r$ is the node corresponding to the root of a fixed nice tree decomposition of $G$.

Let $t\in V(T)$ be the node where $\mu(t)>2k'/3$ and for each child $t'$ of $t$, $\mu(t')\leq 2k'/3$.
Since $\mu$ is a good labeling function, it is easy to see that this node exists and is unique provided that $k'>0$. 
Moreover, observe that $t$ could  either be a forget node or a join node. We
distinguish these two cases.

\begin{itemize}
\setlength{\itemsep}{-1pt}
\item{\em Case 1.} 
If $t$ is  a forget node, we set $V_{1}=V(H_{t'})$ and $V_{2}=V(G)\setminus (V_{1}\cup X_{t'})$
and observe that $P_{\theta_c}(G[V_i])\leq \lfloor 2k'/3\rfloor, i=1,2$. Also we
set $X=X_{t'}$. 

\item{\em Case 2.} If $t$ is a join node with children $t_{1}$ and $t_{2}$, we have that 
$\mu(t_{i})\leq 2k'/3, i=1,2$.  
However, as $\mu(t_{1})+\mu(t_{2})>2k'/3$, we also have that 
either $\mu(t_{1})\geq k'/3$ or $\mu(t_{2})\geq k'/3$. Without loss of generality we assume that 
$\mu(t_{1})\geq k'/3$ and we set $V_{1}=V(H_{t_{1}})$, $V_{2}=V(G)\setminus (V_{1}\cup X_{t_{1}})$
and $X=X_{t_{1}}$.
\end{itemize}
Now we argue that if $G$ has a $\cal F$--hitting set of size at most $k$ then 
then the size of the hitting set returned by the algorithm is upper bounded by $O(k \log^{3/2} k)$. As in 
the first phase we can argue that the size of the union of the sets added to the hitting set in the subproblems at depth $i$ 
is at most   $d'(k+d) \sqrt{\log (k+d)}$. Observe that the recursive procedure in Algorithm~\ref{fig:jchitsetconst3} 
is such that the value of the function $\mu()$ drops by at least a constant fraction at every level of recursion. This
implies that the depth of recursion is upper bounded by $O(\log |Z|)=O(\log k)$. Hence the size of the 
hitting set returned by the algorithm is upper bounded by $O(k \log^{3/2}k)$ whenever $G$ has a  
$\cal F$--hitting set of size at most $k$. Thus if the size of the hitting set returned by 
{\sc Hit-Set-II-}$(G,Z)$ is more than $d'(k+d) \sqrt{\log^{3/2} (k+d)}$,  we return that $G$ does not have a 
$\cal F$--hitting set of size at most $k$. This concludes the proof. 
\end{proof}

\begin{proof}[\bf Proof of Theorem~\ref{thm:fapprox}] Given a graph
  $G$ on $n$ vertices, let $k$ be the minimum positive integer in
  $\{1,\ldots, n\}$ such that Lemma~\ref{thm:epconstructive}
  returns a $\cal F$-hitting set $S$ when applied on $(G,k)$. We
  return this $S$ as an approximate solution. By our choice of $k$
  we know that $G$ does not have $\cal F$-hitting set of size at
  most $k-1$ and hence $OPT\geq k$. This implies that the size of
  $S$ returned by Lemma~\ref{thm:epconstructive} is at most $O(k
  \log^{3/2} k)=O(OPT \log^{3/2} OPT)$.  This concludes the proof.
\end{proof}

We now define a generic problem. Let $\eta$ be a fixed constant. In the {\sc Treewidth $\eta$-Deletion Set}  problem, we are 
given an input graph $G$ and the objective is to delete minimum number of vertices from a graph such that the resulting graph has treewidth at most $\eta$. For an example  {\sc Treewidth $1$-Deletion Set} is simply the {\sc Feedback vertex set} problem. We obtain  the following corollary of Theorem~\ref{thm:fapprox}.

\begin{corollary}
{\sc Feedback Vertex Set}, 
{\sc Diamond Hitting Set}, {\sc Pathwidth One Deletion Set}, {\sc Outerplanar Deletion Set} and 
{\sc Treewidth $\eta$-Deletion Set} 
admit a factor $O(\log^{3/2} n)$ approximation algorithm on general undirected graphs. 
\end{corollary}


\section{Kernelization for \tfd}
\label{sec:kernelization}
In this section we obtain a polynomial kernel for \tfd{} on general graphs. To obtain our 
kernelization algorithm we not only need approximation algorithm presented in the last 
section but also a variation of classical Hall's theorem. We first present this combinatorial tool and other 
auxiliary results that we make use of. 

\subsection{Combinatorial Lemma and some Linear-Time Subroutines.}
We need a variation of the celebrated Hall's Theorem,
which we call the \qel{}. The \qel{} is a generalization of a
result due to Thomass\'e~\cite[Theorem~2.3]{T09}, and captures a
certain property of neighborhood sets in graphs that implicitly
has been used by several authors to obtain polynomial kernels for
many graph problems. For $q=1$, the application of this lemma is
exactly the well-known Crown Reduction Rule~\cite{AFLS07}. 


\paragraph{ The Expansion Lemma.}
Consider a bipartite graph $G$ with vertex bipartition $A\uplus B$. Given 
subsets $S\subseteq A$ and $T\subseteq B$, we say that $S$ has $|S|$ $q$-stars
in $T$ 
if to every $x\in S$ we can associate a subset $F_x\subseteq N(x)\cap T$ such
that (a) for all $x\in S$, $|F_x|=q$; (b) for any pair of vertices $x,y\in S$, 
$F_x\cap F_y=\emptyset$. Observe that if $S$ has $|S|$ $q$-stars in $T$ then 
every vertex $x$ in $S$ could be thought of as the center of a star with its $q$ 
leaves in $T$, with all these stars being vertex-disjoint. Further, a collection of $|S|$ $q$-stars
is also a family of $q$ edge-disjoint matchings, each saturating $S$. 
We use the following result in our kernelization algorithm 
to bound the degrees of vertices. 

\begin{lemma}\label{lem:q-expansion-lemma} {\rm\textbf{{[}The \qel{}{]}}}
Let $q$ be a positive integer, and let $m$ be the size of 
the maximum matching in a bipartite graph $G$ with vertex bipartition $A \uplus B$. If  $|B|> mq $, and there are no isolated
vertices in $B$, then there exist nonempty vertex sets $S\subseteq A,T\subseteq
B$ 
such that  $S$ has $|S|$ $q$-stars in $T$ and no vertex in $T$ has a neighbor outside $S$. Furthermore, the sets 
$S,T$ can be found in time polynomial in the size of $G$.
\end{lemma}
\begin{proof}
Consider the graph $H=(X\uplus B,E)$ obtained from $G=(A\uplus B,E)$ by adding $(q-1)$ copies of all the vertices in $A$, and giving all copies of a vertex $v$ the same neighborhood in $B$ as $v$. Let $M$ be a maximum matching in $H$. In further discussions, vertices are saturated and unsaturated  with respect to this fixed matching $M$.

Let $U_X$  be the vertices in $X$ that are unsaturated, and $R_X$ be those that are reachable from $U_X$ via alternating paths. We let $S_A = X \setminus (U_X \cup R_X)$. Let $U_B$ be the set of unsaturated vertices in $B$, and let $S'$ denote the set of partners of $S_A$ in the matching $M$, that is, $S' = \{x \in B ~|~ \{u,x\} \in M \mbox{ and } u \in S_A\}$. Let $T = S' \cup U_B$ (see Figure~\ref{fig:expansion-lemma-proof}).

\begin{figure}[H]
\begin{center}
 \includegraphics[scale=0.6]{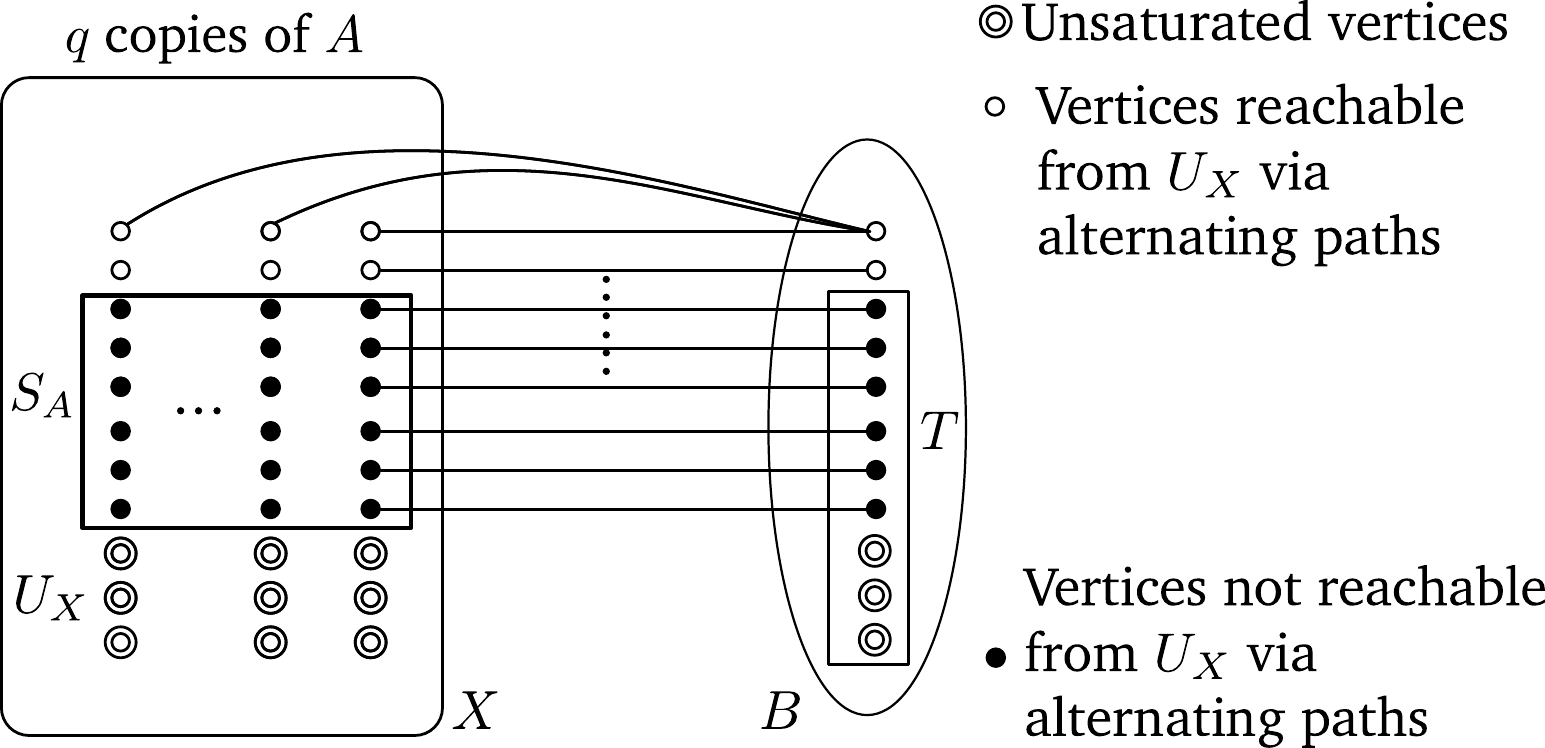}
 \caption{\label{fig:expansion-lemma-proof} The construction used in the proof
 of the \qel{}}
\end{center}
\end{figure}

For every $v \in A$, let $C(v)$ be the set of all copies of $v$ (including $v$). We claim that either $C(v) \cap S_A = C(v)$, or $C(v) \cap S_A = \emptyset$. Suppose that $v \in S_A$ but a copy of $v$, say $u$, is in $U_X$. Let $\{v,w\} \in M$. Then $v$ is reachable from $u$ because $\{u,w\} \in E(H)$, and hence $w$ is not unsaturated in $M$, contradicting the assumption that $w \in U_X$. In the case when $v \in S_A$ but a copy of $u$ is in $R_X$, let $\{w,u\}$ be the last edge on some alternating path from $U_X$ to $u$. Since $\{w,v\} \in E(H)$, we have  that there is also an alternating path from $U_X$ to $v$, contradicting the fact that $v \in S_A$. Let $S=\{v\in A \vert C(v)\subseteq S_A\}$.  Then the subgraph $G[S \cup T]$ contains $q$ edge-disjoint matchings, each of which saturates $S$ in G --- this is because in $H$, $M$ saturates each copy of $v\in S$ separately.

If no vertex in $T$ has a neighbor outside $S_A$ in $H$, then from the construction no vertex in $T$ has a neighbor outside $S$ in $G$. We now prove that no vertex in $T$ has a neighbor outside $S_A$ in $H$. For the purpose of contradiction, let us assume that for some $v \in T$,  $u \in N(v)$, but $u \notin S_A$. First, consider the case when $v \in S'$. Suppose $u \in R_X$. We know that $u \in R_X$ because there is some unsaturated vertex (say $w$) that is connected by an alternating path to $u$. This path can be extended to a path to $v$ using the edge $\{u,v\}$, and can be further extended to $v'$, where $\{v,v'\} \in M$. However, $v' \in S_A$, and by construction, there is no path from $w \in U_X$ to $v'$, a contradiction. If $u \in U_X$, then we arrive at a contradiction along the same lines (in fact, the paths from $w$ to a vertex in $S$ will be of length two in this case). Now consider the case when $v \in U_B$. Again, we may arrive at $u$ from some $w \in U_X$ (if $u \in R_X$) or $\{u,w\}$ is an independent edge outside $M$ (if $u \in U_X$). In both cases, we have an {\em augmenting} path, contradicting the fact that $M$ is a maximum matching. This completes the proof.
\end{proof}

We will need the following proposition  for the proof of next observation. Its 
proof follows from definitions. 
\begin{proposition}\label{cor:min-jc-minor-model}
  For any $c\in\mathbb{N}$, a subgraph $M$ of   graph $G$ is a minimal minor-model of \jc{} in $G$
  if and only if $M$ consists of two trees, say $T_1$ and $T_2$, and a
  set $S$ of $c$ edges, each of which has one end vertex in $T_1$ and
  the other in $T_2$.
\end{proposition}

\begin{observation}
\label{obs:no-cut-pendant-vertex} For $c\ge2$, any minimal \jc{} minor-model $M$ of a graph $G$
is a connected subgraph of $G$, and does not contain a vertex whose degree in $M$ is less than $2$, or a vertex whose deletion from $M$ results in a disconnected graph
(a cut vertex of $M$).
\end{observation}
\begin{proof}
From Proposition~\ref{cor:min-jc-minor-model}, whose terminology we use in
this proof, $M$ is connected
and contains no isolated vertex. Suppose $x$ is a vertex of degree
exactly one in $M$. Then $x$ is a leaf node in one of the two trees in
$M$, say $T_{1}$, and no edge in $S$ is incident on $x$. Removing $x$
from $T_{1}$ results in a smaller \jc{} minor-model, contradicting
the minimality of $M$. It follows that every vertex of $M$ has degree
at least two.

Now suppose $x$ is a cut vertex in $M$ which belongs to, say, the
tree $T_{1}$. Let $T_{1}^{1},T_{1}^{2},\ldots,T_{1}^{l}$ be the
subtrees of $T_{1}$ obtained when $x$ is deleted from $T_{1}$.
Let $M'$ be the graph obtained by deleting $x$ from $M$. If $l>0$,
then each $T_{1}^{i}$ has a leaf node, which, by the above argument,
has at least one neighbor in $T_{2}$. If $l=0$, then $M'=T_{2}$.
Thus $M'$ is connected in all cases, and so $x$ is not a cut vertex,
a contradiction.%
\end{proof}


The following well known result states that every
optimization problem expressible in MSO has a linear time algorithm 
on  graphs of bounded treewidth. 
\begin{proposition}[\cite{ArnborgLagergrenSeese91,Bodlaender96ali,BorieParkerTovey1992,Courcelle1990,CourcelleMosbah1993}]%
\label{prop:lineartimetwalgo}%
Let $\phi$ be a property that is expressible in Monadic Second Order
Logic. For any fixed positive integer $t$, there is an algorithm 
that, given a graph $G$ of treewidth at most $t$ as input,
finds a largest (alternatively, smallest) set $S$ of vertices of
$G$ that satisfies $\phi$ in time $f(t,|\phi|)|V(G)|$. 
\end{proposition}

Proposition~\ref{prop:lineartimetwalgo} together with MSO  formulations~\ref{cmso:jchitset} and~\ref{cmso:jcflower} given in Section~\ref{sec:mso-thetac} implies the following  lemma.

\begin{lemma}
\label{lem:linALgTW}%
Let $G$ be a graph on $n$ vertices and $v$ a vertex of $G$.  Given a 
tree decomposition of width $t \in O(1)$ of $G$, we can, in $O(n)$ time, find both 
(1) a smallest set $S\subseteq V$ of vertices of
$G$ such that the graph $G \setminus S$ does not contain \jc{}
as a minor, and (2) a largest collection 
$\{M_{1},M_{2},\ldots,M_{l}\}$ of \jc{} minor
models of $G$ such that for $1\le i<j\le l,(V(M_{i})\cap V(M_{j}))=\{v\}$.
%
\end{lemma}

Now we describe the reduction rules used by the kernelization
algorithm. In contrast to the reduction rules employed by most
known kernelization algorithms, these rules cannot always be
applied on general graphs in polynomial time. Hence the algorithm
does { not} proceed by applying these rules exhaustively, as is
typical in kernelization programs. We describe how to arrive at
situations where these rules can in fact be applied in polynomial
time, and prove that even this selective application of rules
results in a kernel of size polynomial in the parameter $k$.

\subsection{Bounding the Maximum Degree of a Graph}
Now we present a set of reduction rules which, given an input
instance $(G,k)$ of \tfd{}, obtains an equivalent instance 
$(G',k')$ where $k'\leq k$ and the maximum degree of $G'$ is at
most a polynomial in $k$. In the sequel a vertex $v$ is {\em irrelevant} if it
is not a part of any \jc{} minor model, and is  {\em relevant} otherwise. For
each rule below, the input instance is $(G,k)$.

\begin{redrule}[Irrelevant Vertex Rule]
Delete all irrelevant vertices in $G$. 
\end{redrule}

Given a graph $G$ and a vertex $v\in V(G)$, an {\em $\ell$-flower passing
through $v$} is a set of $\ell$ different \jc{} minor-models in $G$, each
containing $v$ and no two sharing any vertex other than $v$. 

\begin{redrule}[Flower Rule]
If a $(k+1)$-flower passes through a vertex $v$ of $G$, then include $v$ in the 
solution and remove it from $G$ to obtain the equivalent instance $(G\setminus
\{v\},(k-1))$. 
\end{redrule}

The argument for the soundness of these reduction rules is simple and is hence omitted. 
One can test whether a particular vertex $v$ 
is part of any minimal minor-model corresponding to \jc{} using the rooted minor testing algorithm of 
Robertson and Seymour~\cite{RobertsonS13}. It is not clear, however,
that one might check whether a vertex is a part of $(k+1)$-\jc{} flower in polynomial time.
Hence we defer the application of these rules and apply them only when the vertices are ``evidently'' 
irrelevant or finding a flower can be solved in polynomial time. Now we state an auxiliary lemma which will be useful in bounding the maximum degree of the graph. 
\begin{lemma}
\label{lem:specializedhitset}%
Let $G$ be a $n$-vertex graph containing \jc{} as a minor and $v$ be a vertex such that $G'=G\setminus \{v\}$ does not contain \jc{} as a minor and the maximum size of a 
flower containing $v$ is at most $k$. Then there exists a set $T_v$ of size $O(k)$ such that $v\notin T_v$ and $G\setminus T_v$ does not contain \jc{} as a minor. Moreover 
we can find the set $T_v$ in polynomial time.
\end{lemma}
\begin{proof}
We first bound the treewidth of $G'$. Robertson, Seymour and
Thomas~\cite{RobertsonST94} have shown that any graph with
treewidth greater than $20^{2c^5}$ contains a $c\times c$ grid, and
hence \jc{}, as a minor. This implies that for a fixed $c$, $\tw(G') \le 20^{2c^5} = O(1)$. Now we 
show the existence of a $T_v$ of the desired kind. Recall the algorithm used to show the existence of a \jc{} hitting set for a graph described in 
Algorithm~\ref{fig:jchitsetconst3}. We use the same algorithm to construct the desired $T_v$. Let $F_{\theta_c}(G)$ denote the size of the maximum flower passing through 
$v$ in $G$. Consider a nice tree decomposition $(T,\mathcal{ X}=\{X_{t}\}_{t\in V(T)})$ of $G'$ of width at most $\tw(G')$. 
We define the function  $\mu(t) := F_{\theta_c}(G[V(H_t)\cup \{v\}])$. It is easy to see that $\mu$ is a good labeling function, and can be computed in polynomial time due to Lemma~\ref{lem:linALgTW}. 
Observe that $\mu(r) \leq k$, where $r$ is the root node of the tree decomposition. Let ${\cal S}(G',k)$ denote the size of the hitting set returned by the algorithm.  
Thus the size of the hitting set returned by the algorithm {\sc Hit-Set-II} (Algorithm~\ref{fig:jchitsetconst3}) is governed by the following recurrence: 
$$ {\mathcal S}(G',k) \leq \max_{1/3\leq \alpha \leq 2/3}\Big\{ {\mathcal
S}(G[V_1],\alpha k)+{\mathcal S}(G[V_2],(1-\alpha)k)+O(1)\Big\}. $$
Using Akra-Bazzi~\cite{AkraB98} it follows that the above recurrence
solves to $ O(k)$. This implies that there exists a  set
$T_v$ of size $O(k)$ such that $v\notin T_v$ and $G\setminus T_v$ does not contain \jc{} as a minor.
We now proceed to find an optimal hitting set in $G$ avoiding $v$. To make the algorithm {\sc Hit-Set-II} run in polynomial time we only need to find the tree decomposition and compute the function $\mu()$ in polynomial time. Since  $\tw(G) = O(1)$, we can find the desired tree decomposition of $G$ or one of its subgraphs in linear time using the algorithm of Bodlaender~\cite{Bodlaender96ali}. Similarly we can compute a flower of the maximum size using Lemma~\ref{lem:linALgTW} in linear time. Hence the function $\mu()$ can also be computed in polynomial time. This concludes the proof of the lemma.
\end{proof}

\paragraph{Flowers, Expansion and the Maximum Degree.} 
Now we are ready to prove the lemma which bounds the maximum degree of the instance. 
\begin{lemma}%
\label{lem:maxdegreebound}%
There exists a polynomial time algorithm that, given an instance $(G,k)$ of $p$-{\sc \jc{}-Deletion} returns an equivalent instance 
$(G',k')$ such that $k'\leq k$ and that the maximum degree of $G'$ is $O(k \log^{3/2} k)$.  
Moreover it also returns a \jc{}-hitting set of $G'$ of size  $O(k \log^{3/2}k)$.
\end{lemma}
\begin{proof}
Given an instance $(G,k)$ of  $p$-{\sc \jc{}-Deletion}, we first apply 
Lemma~\ref{thm:epconstructive} on $(G,k)$. The polynomial time algorithm described in Lemma~\ref{thm:epconstructive}, 
given a graph $G$ and a positive integer $k$ either reports that $G$ has no \jc{}-hitting set of size at most $k$, or 
finds a \jc{}-hitting set of size at most $k^*=O(k\log^{3/2} k)$. If the algorithm reports that $G$ has has no \jc{}-hitting set 
of size at most $k$, then we return that $(G,k)$ is a {\sc NO}-instance to $p$-{\sc \jc{}-Deletion}. So we assume that 
we have a hitting set $\mathcal{S}$ of size $k^*$. Now we proceed with  the following two rules.

\medskip
\noindent\textbf{Selective Flower Rule.} To apply the Flower Rule selectively we use  $\mathcal{S}$, the 
\jc{}-hitting set. For a vertex $v \in \mathcal{S}$ let  $\mathcal{S}_v := \mathcal{S} \setminus \{v\}$ and let $G_v := G\setminus \mathcal{S}_v$. 
By a result of Robertson et. al.~\cite{RobertsonST94} we
know that any graph of treewidth greater than  $20^{2c^5}$ contains a 
$c\times c$ grid, and hence \jc{}, as a minor. 
Since deleting $v$ from $G_v$ makes it \jc{}-minor-free, $\tw(G_v)\leq  20^{2c^5}+1 = O(1)$. Now by Lemma~\ref{lem:linALgTW},  we find in linear time the size of the largest flower centered at $v$, in $G_v$. If for any vertex $v\in \mathcal{S}$ the size of the flower in $G_v$ is 
at least $k+1$, we apply the Flower Rule and get an equivalent instance  $(G \leftarrow G\setminus \{v\},k \leftarrow k-1)$. Furthermore, we 
set $\mathcal{S}:=\mathcal{S}\setminus \{v\} $. We apply the Flower Rule selectively until no longer possible. We abuse notation and continue to use $(G,k)$ to refer to the instance that is reduced with respect to exhaustive application of the 
Selective Flower Rule. Thus, for every vertex $v \in \mathcal{S}$ the size of any flower passing through $v$ in $G_v$ is  at most $k$.  


Now we describe how to find, for a given  $v \in V(G)$, a hitting set $H_v \subseteq V(G) \setminus \{v\}$ for all minor-models of \jc{} that contain $v$. 
Notice that this hitting set is required to {\em exclude} $v$, so $H_v$ cannot be the trivial hitting set $\{v\}$. 
If $v \notin \mathcal{S}$, then $H_v = \mathcal{S}$. On the other hand, suppose $v \in \mathcal{S}$. 
Since the maximum size of a flower containing $v$ in the graph $G_v$ is at most $k$ by Lemma~\ref{lem:specializedhitset}, we 
can find a set $T_v$ of size $O(k)$ that does not contain  $v$ and hits all the \jc{} minor-models passing through $v$ in $G_v$. Hence in this case 
we set $H_v=\mathcal{S}_v \cup T_v$ (See Figure~\ref{fig:flower-rule}.). 
We denote $|H_v|$ by $h_v$. Notice that $H_v$ is defined algorithmically, that is, there could be many small hitting sets in $V(G) \setminus \{v\}$ hitting all minor-models containing $v$, and $H_v$ is one of them.

\begin{figure}[h]
\begin{center}
 \includegraphics[scale=0.7,bb=0 0 320 180]{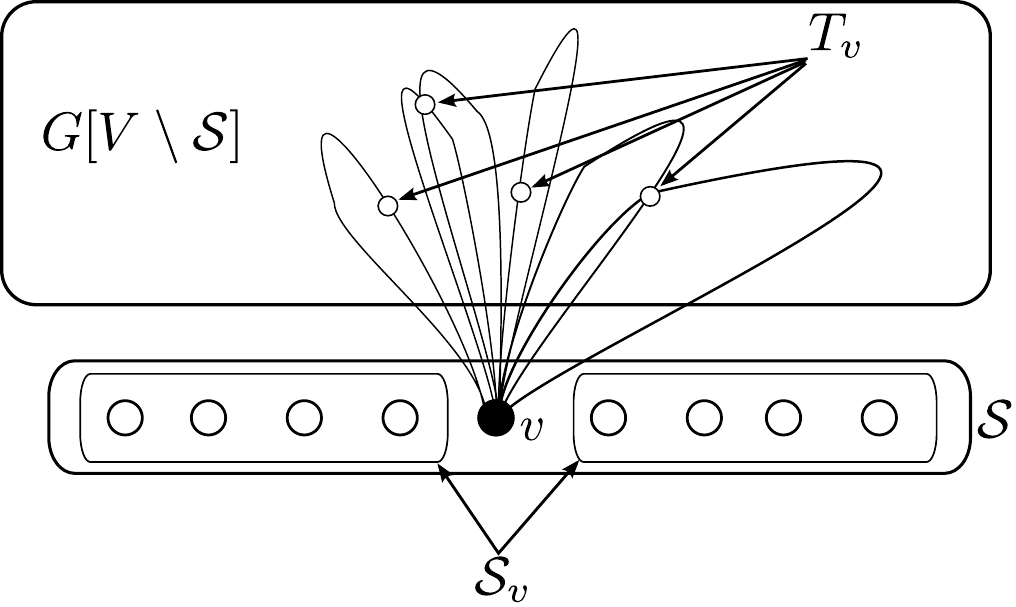}
 \caption{\label{fig:flower-rule} The hitting set in Selective
 Flower Rule}
\end{center}
\end{figure}

\paragraph{$q$-expansion Rule with $q=c$.} Given an instance $(G,k)$, $\cal S$, and a family of sets $H_v$,  we show that 
if there is a vertex $v$ with degree more than $ch_v+c(c-1)h_v$, then we can reduce its degree to at most $ch_v+c(c-1)h_v$ 
by repeatedly  applying the \qel{}  with $q=c$. Observe that for every  vertex $v$ the set $H_v$ is also a \jc{} hitting set for $G$, that is, $H_v$ hits {\em all} minor-models of \jc{} in $G$. Consider the graph $G \setminus H_v$. Let the components of this graph that contain a neighbor of $v$ be $C_1, C_2, \ldots, C_r$. Note that $v$ cannot have more than $(c-1)$ neighbors into any component, else contracting the component will form a \jc{} minor and will 
contradict the fact that $H_v$ hits all the \jc{} minors. 
Also note that none of the $C_i$'s can contain a minor model of \jc{}.

We say that a component $C_i$ is adjacent to $H_v$ if there exists a vertex $u\in C_i$ and $w\in H_v$ such that $(u,w)\in E(G)$. Next we show that  vertices in components that are not adjacent to $H_v$ are irrelevant in $G$. Recall a vertex is irrelevant if there is no minimal minor model of \jc{} that contains  it. Consider a vertex $u$ in a component $C$ that is not adjacent to $H_v$. Since $G[V(C)\cup \{v\}]$ does not contain any \jc{} minor we have that  if $u$ is a part of a minimal minor model $M \subseteq G$, then $v \in M$ and also there exists  a vertex $u'\in M$ such that $u' \notin C \cup \{v\}$. Then the removal of $v$ disconnects $u$ from $u'$ in $M$, a contradiction to Observation~\ref{obs:no-cut-pendant-vertex} 
that for $c\ge2$, any minimal \jc{} minor model $M$ of a graph $G$ does not contain a cut vertex.
Applying the Irrelevant Vertex Rule to the vertices in all such components leaves 
us with a new set of components $D_1, D_2, \ldots , D_s$, such that for every $i$, in $D_i$, there is at least one vertex that is adjacent to a vertex in $H_v$. 

As before, we continue to use $G$ to refer to the graph obtained after the Irrelevant Vertex Rule has been applied in the context described above. We also update the sets $H_v$ for $v\in V(G)$ by deleting all the vertices $w$ from these sets those have been removed using Irrelevant Vertex Rule.

Now, consider a bipartite graph $\cal G$ with vertex bipartitions $H_v$ and $D$. Here $D = \{d_1, \ldots, d_s\}$ contains a vertex $d_i$ corresponding to each component $D_i$. We add an edge $(v,d_i)$ if there is a vertex $w\in D_i$ such that $\{v,w\}\in E(G)$.  
Even though we start with a simple graph (graphs without parallel edges) it is possible that   after applying reduction rules parallel edges may appear. However, throughout the algorithm, we ensure that the number of parallel edges between any pair of vertices is at most $c$. 
Now, $v$ has at most $ch_v$ edges to vertices in $H_v$. Since $v$ has at most $(c-1)$ edges to each $D_i$, it follows that if $d(v) >  ch_v+c(c-1)h_v$, 
then the number of components $|D|$ is more than $ch_v$. 
Now by applying   \qel{} with $q=c$, $A=H_v$, and $B=D$, we  find a subset $S \subseteq H_v$ and $T \subseteq D$ such that $S$ has $|S|$ $c$-stars in $T$ and $N(T)=S$. 

The reduction rule involves deleting edges of the form $(v,u)$ for all $u \in D_i$, such that $d_i \in T$, and adding $c$ edges between $v$ and $w$ for all $w \in S$. We 
add these edges only if they were not present before so that the number of egdes between any pair of vertices remains at most $c$. 
This completes the description of the 
  $q$-expansion reduction rule with $q=c$.   Let $G_R$ be the graph obtained after 
applying the reduction rule.  The following lemma shows the correctness of the rule. 
\begin{lemma}
\label{lem:q-expansioncorrectness}%
Let $G$, $S$ and $v$ be as above and $G_R$ be the graph obtained after applying the $c$-expansion rule. 
Then $(G,k)$ is an yes instance of $p$-{\sc \jc{}-Deletion} if and only if $(G_R,k)$ is an yes instance of $p$-{\sc \jc{}-Deletion}. 
\end{lemma}

\begin{proof}
We first show that if $G_R$ has hitting set $Z$ of size at most $k$, then the same hitting set $Z$ hits all the minor-models of \jc{} in $G$. Observe that either  $v\in Z$ or $S\subseteq Z$. Suppose $v\in Z$, then observe that
$G_R\setminus \{v\}$ is the same as $G\setminus \{v\}$. Therefore  $Z\setminus \{v\}$, a 
hitting set of $G_R\setminus \{v\}$ is also a hitting set of $G\setminus \{v\}$. This shows that $Z$ is a hitting set of size at most $k$ of $G$. The case when $S\subseteq Z$ is similar. 

To prove that a hitting set of size at most $k$ in $G$ implies a hitting set of size at most $k$ in $G_R$, it suffices to prove that whenever there is a hitting set of size at most $k$, there also exists a hitting set of size at most $k$ that contains either $v$ or all of $S$. Consider a hitting set $W$ that does not contain $v$, and omits at least one vertex from $S$. Note the $|S|$ $c$-stars in $\cal{G}[S\cup T]$, along with $v$, correspond to minor-models of \jc{} centered at $v$ in $G$, vertex-disjoint except for $v$. Thus, such a hitting set must pick at least one vertex from one of the components. Let $\mathcal{D}$ be the collection of components $D_i$ such that the (corresponding) vertex $d_i \in T$. Let $X$ denote the set of all vertices of $W$ that appeared in any $D_i \in \mathcal{D}$. Consider the hitting set $W'$ obtained from $W$ by removing $X$ and adding $S$, that is, $W' := (W \setminus X) \cup S$.

We now argue that $W'$ is also a hitting set of size at most $k$. Indeed, let $S'$ be the set of vertices in $S$ that do not already belong to $W$. Clearly, for {\em every} such vertex that $W$ omitted, $W$ must have had to pick distinct vertices from $\mathcal{D}$ to hit the \jc{} minor-models formed by the corresponding $c$-stars. Formally, there exists a $X' \subseteq X$ such that there is a bijection between $S'$ and $X'$, implying that $|W'| \leq |W| \leq k$. 

Finally, observe that $W'$ must also hit all minor-models of \jc{} in $G$. If not, there exists a minor-model $M$ that contains some vertex $u \in X$. Hence, $u \in D_i$ for some $i$, and $M$ contains some vertex in $H_v \setminus S$. However, $v$ separates $u$ from $H_v \setminus S$ in $G \setminus S$, contradicting Observation~\ref{obs:no-cut-pendant-vertex} that $M$ does not contain a cut vertex. This concludes the proof.
\end{proof}

Observe that all edges that are added during the application of the $q$-expansion reduction rule have at least one end point in $\cal{S}$, and hence $\cal{S}$ remains a hitting set of $G_R$. We are now ready to summarize the algorithm that bounds the degree of the graph (see Algorithm~\ref{alg:bdd-degree}).

\begin{algorithm}[h]
\caption{{\sc Bound-Degree}$(G,k,\cal{S})$}
\label{alg:bdd-degree}
\begin{algorithmic}[1]
\STATE Apply the Selective Flower Rule 
\IF {$\exists v \in V(G)$ such that $d(v) >  ch_v+c(c-1)h_v$}
\STATE Apply the $q$-expansion reduction rule with $q=c$. 
\ELSE
\STATE Return $(G,k,\cal{S})$.
\ENDIF
\STATE Return {\sc Bound-Degree}$(G,k,\cal{S})$.
\end{algorithmic}
\end{algorithm}        

Let the instance output by Algorithm \ref{alg:bdd-degree} be  $(G',k',\cal{S})$.
Clearly, in $G'$, the degree of every vertex is at most $ch_v+c(c-1)h_v \leq O(k\log^{3/2} k)$. The routine also returns $\cal S$ --- a \jc{}-hitting set of $G'$ of size at most $O(k \log^{3/2} k)$.

We now show that the algorithm  runs in polynomial time. For $x \in V(G)$, let $\nu(x)$ be the number of neighbors of $x$ to which $x$ has fewer than $c$ parallel edges. Observe that the application of $q$-expansion reduction rule never increases $\nu(x)$ for any vertex and decreases $\nu(x)$ for at least one vertex. The other rules delete vertices, which can never increase $\nu(x)$ for any vertex. This concludes the proof.
\end{proof}

\subsection{Analysis and Kernel Size -- Proof of Theorem~\ref{thm:thetackernel}}
In this section we give the desired kernel for \tfd.
\begin{proof}[\bf Proof of Theorem~\ref{thm:thetackernel}]
Let  $(G,k)$ be an instance to \tfd{}. We first bound the maximum degree of the graph by applying Lemma~\ref{lem:maxdegreebound} on $(G,k)$. If 
Lemma~\ref{lem:maxdegreebound}  returns that $(G,k)$ is a NO-instance to \tfd{} then 
we return the same. Else we obtain an equivalent instance $(G',k')$ such that $k'\leq k$ 
and the maximum degree of $G'$ is bounded by $O(k \log^{3/2} k)$.  Moreover 
it also returns a \jc{}-hitting set, $X$, of $G'$ of size at most $O(k \log^{3/2} k)$.
Let $d$ denote the treewidth of the graph after the removal of $X$, that is, $d := \tw(G \setminus X)$.
 
Now, we obtain our kernel in two phases: we first apply the protrusion rule 
selectively (Lemma~\ref{lem:red2finiteindex}) and get a polynomial kernel. Then, we apply the protrusion 
rule exhaustively on the obtained kernel to get a smaller kernel.  To obtain the kernel we follow the following steps. 
\paragraph{\sl Applying the Protrusion Rule.}
By a result of Robertson et. al.~\cite{RobertsonST94} we
know that any graph of treewidth greater than  $20^{2c^5}$ contains a 
$c\times c$ grid, and hence \jc{}, as a minor. 
Hence $d \leq 20^{2c^5}$. 
Now we apply Lemma~\ref{lem:prottfd} and get a $2(d+1)$-protrusion $Y$ of $G'$ of size at least 
$\frac{|V(G')|-|X|}{4|N(X)|+1}$.  By Lemma~\ref{lem:jcisfii}, \tfd{} has finite integer index. Let 
 $\gamma : \mathbb{N} \rightarrow \mathbb{N}$  be the function defined in Lemma~\ref{lem:red2finiteindex}. 
 Hence if $\frac{|V(G')|-|X|}{4|N(X)|+1}\geq \gamma(2d+1)$ then using Lemma~\ref{lem:red2finiteindex} we 
replace the  $2(d+1)$-protrusion $Y$  of $G'$ and obtain an instance $G^*$ such that  $|V(G^*)| < |V(G')|$, $k^* \leq k'$, and $(G^*,k^*) $ is a YES-instance of \tfd{} if and only if $(G',k') $ is a YES-instance of \tfd{} . 

Before applying the Protrusion Rule again, if necessary, we bound 
the maximum degree of the graph by reapplying Lemma~\ref{lem:maxdegreebound}. This is done because the application of the protrusion rule could potentially increase the maximum degree of the graph. 
We alternately apply the protrusion rule and Lemma~\ref{lem:maxdegreebound} in this fashion, until either Lemma~\ref{lem:maxdegreebound} returns that $G$ is a NO instance, or the protrusion rule ceases to apply.
Observe that this process will 
always terminate as the procedure that bounds the maximum degree never increases the number of vertices 
and the protrusion rule always reduces the number of vertices. 

Let $(G^*,k^*)$ be a reduced instance with hitting set $X$. In other words, there is no
$(2d+2)$-protrusion of size $\gamma(2d+2)$ in $G^*\setminus X$, and the protrusion rule no longer applies.
Now we show that the 
number of vertices  and edges of this graph is bounded by $O(k^2 \log^3 k)$. We first bound the number of 
vertices. Since we cannot apply the Protrusion Rule, $\frac{|V(G^*)|-|X|}{4|N(X)|+1}
\leq \gamma(2d+2)$.  Since $k^*\leq k$ this implies that 
\begin{eqnarray*}
|V(G^*)|& \leq & \gamma(2d+2)(4|N(X)|+1)+|X|\\
& \leq &\gamma(2d+2)(4|X|\Delta(G^*)+1)+|X| \\
& \leq & \gamma(2d+2)(O(k \log^{3/2} k) \times O(k \log^{3/2} k)+1)+O(k \log^{3/2} k)\\
& \leq & O(k^2 \log^{3} k).
\end{eqnarray*}
To get the desired bound on the number of edges we first observe that since 
$\tw(G^* \setminus X) \leq 20^{2c^5}=d$, we have that the number of edges in 
$G^* \setminus X \leq d |V(G^*) \setminus X|=O(k^2 \log^3 k).$  Also the number of edges 
incident on the vertices in $X$ is at most  $|X|\cdot \Delta(G^*)\leq O(k^2 (\log k)^{3})$.  This gives us 
a polynomial time algorithm that returns a kernel of size  $O(k^2 \log^3 k)$. 

Now we give a kernel of smaller size. To do so we apply combination of rules to bound the degree and the protrusion rule as before. The only difference is that we would like to replace any large $(2d+2)$-protrusion in graph by a smaller one.   
We find a $2d+2$-protrusion $Y$  of size at least $\gamma(2d+2)$ by guessing the boundary $\partial(Y)$ of size at 
most $2d+2$. This could be performed in time $k^{O(d)}$.  So let $(G^*,k^*)$ be the reduced instance on which we can not apply the Protrusion Rule. Then we know that $\Delta(G^*)=O(k \log^{3/2} k)$. If $G$ is a YES-instance then there exists a \jc{}-hitting set $X$ of size at most $k$ such that $\tw(G \setminus X) \leq 20^{2c^5}=d$. Now applying the analysis above with this $X$  yields that $|V(G^*)|=O(k^2 \log^{3/2} k)$ and $|E(G^*)|\leq O(k^2 \log^{3/2} k)$. 
Hence if the number of vertices or edges in the reduced instance $G^*$, to which we can not apply the Protrusion Rule, 
is more than  $O(k^2 \log^{3/2}k)$  then we return that $G$ is a NO-instance. This concludes the proof of the theorem.
\end{proof}

Theorem~\ref{thm:thetackernel} has following immediate corollary. 
\begin{corollary}
{\sc $p$-Vertex Cover} , {\sc $p$-Feedback Vertex Set}  and {\sc $p$-Diamond Hitting Set}  have kernel of size 
$O(k^2 \log^{3/2}k)$. 
\end{corollary}

 \section{Conclusion}\label{section:conclision}
 In this paper we gave the first kernelization algorithms for a
 subset of \fd{} problems and a generic approximation algorithm
 for the \fd{} problem when the set of excluded minors  $\cal
 F$ contains at least one planar graph. Our approach generalizes
 and unifies known kernelization algorithms for {\sc $p$-Vertex
   Cover} and {\sc $p$-Feedback Vertex Set}.  By the celebrated
 result of Robertson and Seymour, every \fd{} problem is FPT and
 our work naturally leads to the following question: does every
 \fd{} problem have a polynomial kernel?  Can it be that for some
 finite sets of minor obstructions $\mathcal{F}=\{O_1,\dots,
 O_p\}$   the answer to this
 question is NO?  Even the case
 $\mathcal{F}=\{K_5, K_{3,3}\}$, vertex deletion to  planar
 graphs, is an interesting challenge.  Another interesting
 question is if our techniques can be extended to another
 important case  when   $\mathcal{F}$ contains a planar graph.


\end{document}